\documentclass[journal,transmag]{IEEEtran}
\usepackage{amsmath}
\usepackage{color}
\usepackage{cases}
\usepackage{amsfonts,amsmath,amssymb}
\usepackage{mathrsfs}
\usepackage{cite}
\usepackage{graphicx}
\usepackage{url}
\usepackage{enumerate}
\usepackage{subfigure}
\usepackage{hyperref}
\usepackage{algorithm}
\usepackage{algorithmic}
\usepackage{amsmath,bm}
\usepackage{caption}
\usepackage{float}
\usepackage{stfloats}
\usepackage{multirow}
\usepackage{url}
\usepackage{breakurl}
\usepackage{txfonts}
\usepackage{booktabs}
\usepackage{threeparttable}
\usepackage{multirow}
\usepackage{fancyhdr}
\usepackage{lettrine}

\usepackage{geometry}
\DeclareMathOperator*{\argmin}{arg\,min}

\newtheorem{theorem}{Theorem}
\newtheorem{proposition}{Proposition}
\newtheorem{lemma}{Lemma}
\newtheorem{corollary}{Corollary}
\newtheorem{example}{Example}
\newtheorem{remark}{Remark}
\newtheorem{definition}{Definition}

\begin{document}
\title{Exact and Robust Reconstructions of Integer Vectors Based on Multidimensional Chinese Remainder Theorem (MD-CRT)}

\author{\mbox{Li Xiao, Xiang-Gen Xia,~\IEEEmembership{Fellow,~IEEE}, and Yu-Ping Wang,~\IEEEmembership{Senior Member,~IEEE}
}
\thanks{This work was supported in part by NIH under Grants R01GM109068 and R01MH104680, and in part by NSF under Grant 1539067.
}
\thanks{L. Xiao and Y.-P. Wang are with the Department of Biomedical Engineering,
Tulane University, New Orleans, LA 70118, USA (e-mail: lxiao1@tulane.edu; wyp@tulane.edu).}
\thanks{X.-G. Xia is with the Department of Electrical and Computer Engineering,
University of Delaware, Newark, DE 19716, USA (e-mail: xxia@ee.udel.edu).}
}
\maketitle

\thispagestyle{fancy}
\fancyhead{}
\lhead{}
\lfoot{\footnotesize{This work has been submitted to the IEEE for possible publication. Copyright may be transferred without notice, after which this version may no longer be accessible.}}
\cfoot{}
\rfoot{}

\begin{abstract}
The robust Chinese remainder theorem (CRT) has been recently proposed for robustly reconstructing a large nonnegative integer from erroneous remainders. It has found many applications in signal processing, including phase unwrapping and frequency estimation under sub-Nyquist sampling. Motivated by the applications in multidimensional (MD) signal processing, in this paper we propose the MD-CRT and robust MD-CRT for integer vectors. Specifically, by rephrasing the abstract CRT for rings in number-theoretic terms, we first derive the MD-CRT for integer vectors with respect to a general set of integer matrix moduli, which provides an algorithm to uniquely reconstruct an integer vector from its remainders, if it is in the fundamental parallelepiped of the lattice generated by a least common right multiple of all the moduli. For some special forms of moduli, we present explicit reconstruction formulae. Moreover, we derive the robust MD-CRT for integer vectors when the remaining integer matrices of all the moduli left divided by their greatest common left divisor (gcld) are pairwise commutative and coprime. Two different reconstruction algorithms are proposed, and accordingly, two different conditions on the remainder error bound for the reconstruction robustness are obtained, which are related to a quarter of the minimum distance of the lattice generated by the gcld of all the moduli or the Smith normal form of the gcld.
\end{abstract}


\section{Introduction}
The Chinese remainder theorem (CRT) is one of the most fundamental theorems in number theory, and has a long history going back to the 3rd--5th centuries AD \cite{crt_integer,crt_integer1,crt_integer2}. Basically, the CRT allows to uniquely reconstruct a large nonnegative integer from its remainders with respect to a set of small moduli, if the large integer is less than the least common multiple (lcm) of all the moduli. To date, there has been a surge in work on applying the CRT for partitioning a large task into a number of smaller but independent subtasks, which can be performed in parallel. For example, the CRT has been intensively utilized in the signal processing community in the context of cyclic convolution \cite{conv1,conv2}, fast Fourier transform \cite{fft1,fft2}, coprime sensor arrays \cite{coprime1,coprime2,conglin1,conglin2}, to name a few. It also finds applications in various other fields, such as computer arithmetic based on modulo operations (e.g., multiplication of very large numbers), coding theory (e.g., residue number system codes), and cryptography (e.g., secret sharing); see \cite{crt_integer,crt_integer1,crt_integer2} and references therein.

Motivated by the applications of the CRT in phase unwrapping and frequency estimation under sub-Nyquist sampling, a robust remaindering problem has been raised and investigated in \cite{xiaxianggenxia,gangli1,yiminzhang,zihuiyuan,radeee,xiaoweili,wenjiewang,jiaxu1}. In these applications, signals are usually subject to noise, and thereby the detected remainders may be erroneous. Two significant questions underlying the robust remaindering problem are: 1) what is the reconstruction range of the large nonnegative integer? and 2) how large can the remainder errors be to ensure the robust reconstruction?
It is well-known that the CRT is not robust against remainder errors, i.e., a small error in a remainder may result in a large error in the reconstruction solution. Directly applying the CRT to these applications will thus yield poor performance. Recently, the robust CRT has been proposed in \cite{xiaxianggenxia,xiaoweili,wenjiewang} and further systematically studied in \cite{guangwuxu,yangyang,yangyang2,xiaoli1,xiaoli2}, for solving the robust remaindering problem. The robust CRT demonstrates that even though every remainder has a small error, a large nonnegative integer can be robustly reconstructed in the sense that the reconstruction error is upper bounded by the bound of the remainder errors.
Beyond these applications aforementioned, the robust CRT may have or has offered applications in multi-wavelength optical measurement \cite{optics1,optics2,optics3}, distance or velocity ambiguity resolution \cite{wenchaoli,akhlaq,ocean1,ocean2}, fault-tolerant wireless sensor networks \cite{wsn1,wsn2,wsn3}, error-control neural coding \cite{biology1,biology2,biology3}, signal recovery using multi-channel modulo samplers \cite{lugan}, etc. Note that the (robust) CRT has been generalized to (robustly) reconstruct multiple large nonnegative integers from their unordered remainder sets as well \cite{xialao1,xialao2,xiaopingli1,huiyongliao,xiaoli5,hanshen1,hanshen2}.
A thorough review of the robust CRT can be found in \cite{xiaoli}.

In this paper, we extend the CRT and robust CRT for integers to the multidimensional (MD) case, called the MD-CRT and robust MD-CRT for integer vectors, so that they can be utilized in MD signal processing. Note that MD signal processing here refers to true (nonseparable) MD signal processing, since separable MD signal processing is straightforward by handling their 1-dimensional counterparts separately along each dimension. First, through rephrasing the abstract CRT for rings in number-theoretic terms, we derive the MD-CRT for integer vectors with respect to a general set of moduli (namely a set of arbitrary nonsingular integer matrices). It is basically that given a set of nonsingular moduli $\{\textbf{M}_i\}_{i=1}^{L}$, an integer vector $\textbf{m}\in\mathcal{N}(\textbf{R})$ can be uniquely reconstructed from its remainders $\textbf{r}_i$ for $1\leq i\leq L$, where $\textbf{R}$ is a least common right multiple of all the moduli, and $\mathcal{N}(\textbf{R})$ denotes the set of all integer vectors in the fundamental parallelepiped of the lattice generated by $\textbf{R}$. A reconstruction algorithm is proposed as well. Notably, the MD-CRT for integer vectors was previously investigated in \cite{jiawenxian,PPV1} for a special case when the $L$ moduli are given by $\textbf{M}_i=\textbf{U}\bm{\Lambda}_i\textbf{U}^{-1}$ for $1\leq i\leq L$ with $\textbf{U}$ being a unimodular matrix and $\bm{\Lambda}_i$'s coprime diagonal integer matrices. For some other special forms of moduli, we further obtain explicit reconstruction formulae of the MD-CRT for integer vectors in this paper.

Moreover, we derive the robust MD-CRT for integer vectors when the $L$ nonsingular moduli are in the form of $\textbf{M}_i=\textbf{M}\bm{\Gamma}_i$ for $1\leq i\leq L$, where
$\textbf{M}$ is an arbitrary integer matrix, and $\bm{\Gamma}_i$'s are pairwise commutative and coprime integer matrices. As in the robust CRT for integers \cite{xiaxianggenxia,xiaoweili,wenjiewang,guangwuxu,yangyang,yangyang2,xiaoli1,xiaoli2},
we attempt to accurately determine all the folding vectors $\textbf{n}_i$'s (i.e., the quotient vectors of $\textbf{m}$ left divided by the moduli), and
a robust reconstruction of $\textbf{m}$ can be calculated as the average of the reconstructions obtained from the folding vectors, i.e., $\widetilde{\textbf{m}}=\frac{1}{L}\sum_{i=1}^L(\textbf{M}_{i}\textbf{n}_{i}+\widetilde{\textbf{r}}_{i})$, where
$\widetilde{\textbf{r}}_{i}$ denotes the $i_0$-th erroneous remainder. We find that the size of the remainder error bound for the reconstruction robustness
depends on the reconstruction algorithm. In other words, different reconstruction algorithms will lead to different conditions on the remainder error bound.
We then propose two different reconstruction algorithms, and accordingly, we obtain two different conditions on the remainder error bound
for the reconstruction robustness, which are related to a quarter of the minimum distance of the lattice generated by $\textbf{M}$ or the Smith normal form of $\textbf{M}$. At the end, we verify the robust MD-CRT for integer vectors by numerical simulations and apply it to MD frequency estimation when a complex MD sinusoidal signal is undersampled using multiple sub-Nyquist sampling matrices.

The rest of this paper is organized as follows. In Section \ref{sec2}, we recall some background knowledge needed to make this paper more self-contained.
In Section \ref{sec3}, we derive the MD-CRT for integer vectors with respect to a general set of moduli, and provide explicit reconstruction formulae when the moduli are in some
special forms. In Section \ref{sec4}, we investigate the robust MD-CRT for integer vectors, and propose two different algorithms for robust reconstruction, resulting in two different conditions on the remainder error bound for the reconstruction robustness. In Section \ref{sec5}, we present simulation results of the robust MD-CRT for integer vectors as well as its application to MD sinusoidal frequency estimation with multiple sub-Nyquist samplings.
We conclude this paper in Section \ref{sec6}.

\textit{Notations}: Capital and lowercase boldfaced letters are used to denote matrices and vectors, respectively. Let $\mathbb{R}$ and $\mathbb{Z}$ denote the sets of reals and integers, respectively. The transpose, inverse, inverse transpose,
and determinant of a matrix $\textbf{A}$ are denoted as $\textbf{A}^T$, $\textbf{A}^{-1}$, $\textbf{A}^{-T}$, and $\text{det}(\textbf{A})$, respectively.
 Given a set of scalars $a_1, a_2, \cdots, a_D$, we denote by $\text{diag}(a_1,a_2,\cdots,a_D)$ the diagonal matrix with $a_i$ being the $i$-th diagonal element. A $D$-dimensional vector $\textbf{a}\in[c,d)^D$ means that every element of $\textbf{a}$ is in the range of $[c,d)$ and $c,d\in\mathbb{R}$. We denote the $(i,j)$-th element of a matrix $\textbf{A}$ as $A(i,j)$, and the $i$-th element of a vector $\textbf{a}$ as $a(i)$.
The symbols $\textbf{I}$ and $\textbf{0}$ denote the identity matrix and the all-zero vector/matrix, respectively, with size determined from context. The relative complement of a set $\mathcal{A}$ with respect to a set $\mathcal{B}$ is written as $\mathcal{B}\backslash\mathcal{A}$. Throughout this paper, all matrices are square matrices unless otherwise stated.

\section{Preliminaries}\label{sec2}
The preliminary knowledge involved in this paper is mainly related to some fundamental properties in elementary number theory. In this section, we recall general concepts and notations for integer vectors and integer matrices \cite{PPV3,PPV4,smith,lattice1,smith3}.

\begin{itemize}
  \item[\romannumeral1)] \textit{Unimodular matrix}: A matrix $\textbf{U}$ is unimodular if it is an integer matrix and $|\text{det}(\textbf{U})|=1$. For any unimodular matrix $\textbf{U}$, its inverse $\textbf{U}^{-1}$ is also unimodular because of $\textbf{U}^{-1}=\text{adj}(\textbf{U})/\text{det}(\textbf{U})$ and $|\text{det}(\textbf{U}^{-1})|=|\text{det}(\textbf{U})|=1$, where \text{adj}(\textbf{U}) stands for the adjugate of $\textbf{U}$ and is an integer matrix.

  \item[\romannumeral2)] \textit{Divisor}: An integer matrix $\textbf{A}$ is a left divisor of an integer matrix $\textbf{M}$ if $\textbf{A}^{-1}\textbf{M}$ is an integer matrix. Similarly, $\textbf{A}$ is a right divisor of $\textbf{M}$ if $\textbf{M}\textbf{A}^{-1}$ is an integer matrix.

   \item[\romannumeral3)] \textit{Multiple}: A nonsingular integer matrix $\textbf{A}$ is a left multiple of an integer integer $\textbf{M}$ if $\textbf{A}=\textbf{P}\textbf{M}$ for some integer matrix $\textbf{P}$. Similarly, $\textbf{A}$ is a right multiple of  $\textbf{M}$ if $\textbf{A}=\textbf{M}\textbf{Q}$ for some integer matrix $\textbf{Q}$.

  \item[\romannumeral4)] \textit{Greatest common divisor (gcd)}: An integer matrix $\textbf{A}$ is a common left divisor (cld) of $L$ ($L\geq2$) integer matrices $\textbf{M}_1, \textbf{M}_2, \cdots, \textbf{M}_L$, if $\textbf{A}^{-1}\textbf{M}_i$ is an integer matrix for each $1\leq i\leq L$.
       We call $\textbf{B}$ a greatest common left divisor (gcld) of $\textbf{M}_1, \textbf{M}_2, \cdots, \textbf{M}_L$, if any other cld is a left divisor of $\textbf{B}$. Note that among all cld's, a gcld has the greatest absolute determinant and is unique up to postmultiplication by a unimodular matrix (because if $\textbf{B}$ is a gcld, so will be $\textbf{B}\textbf{U}$ for any unimodular matrix $\textbf{U}$). Similarly, a common right divisor (crd) and a greatest common right divisor (gcrd) of $\textbf{M}_1, \textbf{M}_2, \cdots, \textbf{M}_L$ can be defined, respectively.

  \item[\romannumeral5)] \textit{Least common multiple (lcm)}: A nonsingular integer matrix $\textbf{A}$ is a common left multiple (clm) of $L$ ($L\geq2$) integer matrices $\textbf{M}_1, \textbf{M}_2, \cdots, \textbf{M}_L$, if $\textbf{A}=\textbf{P}_i\textbf{M}_i$ for some integer matrix $\textbf{P}_i$ and each $1\leq i\leq L$.
      We call $\textbf{C}$ a least common left multiple (lclm) of $\textbf{M}_1, \textbf{M}_2, \cdots, \textbf{M}_L$, if any other clm is a left multiple of $\textbf{C}$. Note that among all clm's, an lclm has the smallest absolute determinant and is unique up to premultiplication by a unimodular matrix (because if $\textbf{C}$ is an lclm, so will be \textbf{U}$\textbf{C}$ for any unimodular matrix $\textbf{U}$). Similarly, a common right multiple (crm) and a least common right multiple (lcrm) of $\textbf{M}_1, \textbf{M}_2, \cdots, \textbf{M}_L$ can be defined, respectively.

  \item[\romannumeral6)] \textit{Coprimeness}: Two integer matrices $\textbf{M}$ and $\textbf{N}$ are said to be left (right) coprime if their gcld (gcrd) is unimodular. In other words, $\textbf{M}$ and $\textbf{N}$ are left (right) coprime if they have no cld's (crd's) other than unimodular matrices.
\end{itemize}

Note that both divisors and multiples above are always taken to be nonsingular integer matrices in this paper.
Given a $D\times D$ nonsingular integer matrix $\textbf{M}$, we define $\mathcal{N}(\textbf{M})$ by
\begin{equation}
\mathcal{N}(\textbf{M})=\{\textbf{k}\,|\,\, \textbf{k}=\textbf{M}\textbf{x}, \textbf{x}\in[0,1)^{D},\text{ and }\textbf{k}\in\mathbb{Z}^D\}.
\end{equation}
The number of elements in $\mathcal{N}(\textbf{M})$ is equal to $|\text{det}(\textbf{M})|$ \cite{smith3}. In the $1$-dimensional case (i.e., $D=1$), letting $M$ be a positive integer, we have $\mathcal{N}(M)=\{0, 1, \cdots, M-1\}$.

Then, the integer vector division is defined as follows. A $D$-dimensional integer vector $\textbf{m}$ has a unique representation with respect to a $D\times D$ nonsingular integer matrix $\textbf{M}$ as $\textbf{m}=\textbf{M}\textbf{n}+\textbf{r}$, or equivalently
\begin{equation}\label{remainder}
\textbf{m}\equiv \textbf{r} \!\!\mod \textbf{M},
\end{equation}
with $\textbf{r}\in\mathcal{N}(\textbf{M})$, where $\textbf{M}$ is viewed as a modulus, and integer vectors $\textbf{n}$ and $\textbf{r}$ are the folding vector and remainder of $\textbf{m}$ with respect to the modulus $\textbf{M}$, respectively. For simplicity, we write $\textbf{r}$ in (\ref{remainder}) as $\textbf{r}=\langle\textbf{m}\rangle_{\textbf{M}}$.
We can compute $\textbf{r}$ by
\begin{equation}\label{comput1}
\textbf{r}=\textbf{m}-\textbf{M}\lfloor\textbf{M}^{-1}\textbf{m}\rfloor,
\end{equation}
i.e., the folding vector $\textbf{n}$ is computed by $\textbf{n}=\lfloor\textbf{M}^{-1}\textbf{m}\rfloor$, where $\lfloor\cdot\rfloor$ denotes the floor operation that is performed on every element of the vector. Since $\textbf{M}^{-1}$ is in general a matrix with rational elements, $\lfloor\textbf{M}^{-1}\textbf{m}\rfloor$ is subject to round-off errors due to finite precision arithmetic. To this end, an alternative \cite{PPV3} to compute $\textbf{r}$ is given by
\begin{equation}\label{comput2}
\textbf{r}=\textbf{M}\left(\text{adj}(\textbf{M})\textbf{m}\!\!\mod \text{det}(\textbf{M})\right)/\text{det}(\textbf{M}),
\end{equation}
where the modulo operation is performed on every element of $\text{adj}(\textbf{M})\textbf{m}$.

It is well known that when the involved matrices in MD signal processing are diagonal, most results in the $1$-dimensional case can be straightforwardly extended to the MD case by handling their 1-dimensional counterparts separately. For example, when $\textbf{M}$ in (\ref{remainder}) is diagonal, i.e., $\textbf{M}=\text{diag}(M_1,M_2,\cdots,M_D)$, then (\ref{remainder}) is equivalent to
$m(i)\equiv r(i)\!\!\mod M_i$ for $1\leq i\leq D$,
where $m(i)$ and $r(i)$ denote the $i$-th elements of $\textbf{m}$ and $\textbf{r}$, respectively. The division for integer vectors is therefore reduced to that for integers. However, the involved matrices are usually nondiagonal, and extending the results of $1$-dimensional signal processing to the MD case will become nontrivial. The Smith normal form, as a popular tool to diagonalize an integer matrix, has been widely used to simplify several MD signal processing problems;
see, for example, \cite{smith2,smith3}.

\begin{proposition}[The Smith normal form \cite{smith}]\label{pr2}
A $D\times K$ integer matrix $\textbf{M}$ can be decomposed as
\begin{equation}\label{smithform}
\textbf{U}\textbf{M}\textbf{V}=
\begin{cases}
    \left(
      \begin{array}{cc}
        \bm{\Lambda} & \bm{0} \\
      \end{array}
    \right) & \text{if } K>D,\\
        \bm{\Lambda}
       & \text{if } K=D,\\
    \left(
      \begin{array}{c}
        \bm{\Lambda} \\
        \bm{0} \\
      \end{array}
    \right)
                 & \text{if } K<D,
\end{cases}
\end{equation}
where $\textbf{U}$ and $\textbf{V}$ are $D\times D$ and $K\times K$ unimodular matrices, respectively, and $\bm{\Lambda}$ is a $\text{min}(K,D)\times\text{min}(K,D)$ diagonal integer matrix, i.e., $\bm{\Lambda}=\text{diag}(\lambda_1,\lambda_2,\cdots,\lambda_\gamma,0,\cdots,0)$ with $\lambda_i$'s being positive integers and $\gamma$ being the rank of $\textbf{M}$. Also, $\lambda_i$'s satisfy $\lambda_i|\lambda_{i+1}$, i.e., $\lambda_i$ divides $\lambda_{i+1}$, for each $1\leq i\leq\gamma-1$.
Under the conditions, $\bm{\Lambda}$ is unique for a given matrix $\textbf{M}$, while $\textbf{U}$ and $\textbf{V}$ are generally not.
Moreover, $\lambda_i$'s are called the invariant factors and can be computed by $\lambda_i=d_i/d_{i-1}$ for $1\leq i\leq \gamma$, where $d_i$ is the gcd of all $i\times i$ determinantal minors of $\textbf{M}$ and $d_0=1$.
\end{proposition}

\begin{proposition}[The Bezout's theorem \cite{PPV4}]\label{pr3}
Let $\textbf{L}$ be a gcld of integer matrices $\textbf{M}$ and $\textbf{N}$. Then, there exist integer matrices $\textbf{P}$ and $\textbf{Q}$ such that
\begin{equation}\label{btl}
\textbf{M}\textbf{P}+\textbf{N}\textbf{Q}=\textbf{L}.
\end{equation}
Similarly, let $\textbf{L}$ be a gcrd of $\textbf{M}$ and $\textbf{N}$. Then, there exist integer matrices $\textbf{P}$ and $\textbf{Q}$ such that
\begin{equation}\label{btl2}
\textbf{P}\textbf{M}+\textbf{Q}\textbf{N}=\textbf{L}.
\end{equation}
\end{proposition}

In Appendix \ref{ap0}, we introduce how to calculate a gcld $\textbf{L}$ of two given nonsingular $D\times D$ integer matrices $\textbf{M}$ and $\textbf{N}$, and the accompanying $\textbf{P}$ and $\textbf{Q}$ in (\ref{btl}) in the Bezout's theorem; see \cite{PPV4} for details. Similarly, we can calculate a gcrd $\textbf{L}$ of $\textbf{M}$ and $\textbf{N}$, and the accompanying $\textbf{P}$ and $\textbf{Q}$ in (\ref{btl2}).

\begin{proposition}[\!\!\cite{PPV3}]\label{pr4}
Let $\textbf{M}$ and $\textbf{N}$ be two nonsingular integer matrices. When $\textbf{M}\textbf{N}=\textbf{N}\textbf{M}$, the following four statements are equivalent:
$1)$ $\textbf{M}$ and $\textbf{N}$ are right coprime;
$2)$ $\textbf{M}$ and $\textbf{N}$ are left coprime;
$3)$ $\textbf{M}\textbf{N}$ is an lcrm of $\textbf{M}$ and $\textbf{N}$; and
$4)$ $\textbf{M}\textbf{N}$ is an lclm of $\textbf{M}$ and $\textbf{N}$.
\end{proposition}

\begin{remark}
As stated in Proposition \ref{pr4}, when $\textbf{M}$ and $\textbf{N}$ are commutative, i.e., $\textbf{M}\textbf{N}=\textbf{N}\textbf{M}$, their left coprimeness and right coprimeness can imply each other, so we use the simpler term ``\textit{coprimeness}''. Similarly, when $\textbf{M}$ and $\textbf{N}$ are commutative and coprime, their product $\textbf{M}\textbf{N}$ is both an lcrm and an lclm, so we use the simpler term ``\textit{lcm}''.
For the $1$-dimensional case (i.e., integer case), Propositions \ref{pr3} and \ref{pr4} are well-known facts.
\end{remark}

Given a $D\times D$ nonsingular matrix $\textbf{M}$ (which is not necessarily an integer matrix), the set of all integer linear combinations of the columns of $\textbf{M}$, i.e.,
\begin{equation}
\text{LAT}(\textbf{M})=\left\{\textbf{M}\textbf{n}\,|\,\, \textbf{n} \text{ is an integer vector}\right\},
\end{equation}
is called the $D$-dimensional lattice generated by $\textbf{M}$, denoted as $\text{LAT}(\textbf{M})$.
The fundamental parallelepiped of $\text{LAT}(\textbf{M})$ is defined as the region:
\begin{equation}
\mathcal{F}_{\text{LAT}(\textbf{M})}=\left\{\textbf{M}\textbf{x}\,|\,\, \textbf{x}\in[0,1)^{D}\right\}.
\end{equation}
The shape of $\mathcal{F}_{\text{LAT}(\textbf{M})}$ defined above depends on the generating matrix $\textbf{M}$. All lattice cells of $\text{LAT}(\textbf{M})$ have the same volume equal to $|\text{det}(\textbf{M})|$ \cite{PPV4}. One can observe that $\mathcal{F}_{\text{LAT}(\textbf{M})}$ and its shifted copies (i.e., the other lattice cells) constitute the whole real vector space $\mathbb{R}^{D}$.
When $\textbf{M}$ is a nonsingular integer matrix,
we obtain $\mathcal{N}(\textbf{M})\subset\mathcal{F}_{\text{LAT}(\textbf{M})}$ and
$\mathcal{N}(\textbf{M})=\mathcal{F}_{\text{LAT}(\textbf{M})}\cap\mathbb{Z}^{D}$.

\begin{proposition}[\!\!\cite{lattice1}]\label{pr5}
Two nonsingular integer matrices $\textbf{M}$ and $\textbf{N}$ generate the same lattice, i.e., $\text{LAT}(\textbf{M})=\text{LAT}(\textbf{N})$, if and only if $\textbf{M}=\textbf{N}\textbf{P}$, where $\textbf{P}$ is a unimodular matrix.
\end{proposition}

\begin{proposition}[\!\!\cite{lattice1}]\label{pr6}
Given two nonsingular integer matrices $\textbf{M}$ and $\textbf{N}$, let $\textbf{C}$ be an lcrm of $\textbf{M}$ and $\textbf{N}$. Then, $\text{LAT}(\textbf{C})=$ $\text{LAT}(\textbf{M})\cap\text{LAT}(\textbf{N})$.
\end{proposition}

\section{MD-CRT for Integer Vectors}\label{sec3}
The well-known CRT for integers allows the reconstruction of a large nonnegative integer from its remainders with respect to a general set of moduli (namely a set of arbitrary positive integers), and it has been successfully applied in $1$-dimensional signal processing, cryptography, parallel arithmetic computing, coding theory, etc.; see \cite{crt_integer,crt_integer1,crt_integer2} and references therein. In this section, as a natural extension of the CRT for integers, the MD-CRT for integer vectors is systematically studied, which provides a reconstruction algorithm for an integer vector from its remainders with respect to a general set of moduli (namely a set of arbitrary nonsingular integer matrices), and possesses potential usefulness in MD signal processing. To begin with, we briefly revisit the CRT for integers as follows.

\begin{proposition}[CRT for integers \cite{crt_integer2}]\label{crt_integer}
Given $L$ moduli $M_i$ for $1\leq i\leq L$, which are arbitrary positive integers, let
$R$ be their lcm. For an integer $m\in\mathcal{N}(R)$ (i.e., $0\leq m<R$), we can uniquely reconstruct $m$ from its remainders $r_i=\langle m\rangle_{M_i}$ as
\begin{equation}
m=\left\langle\sum_{i=1}^{L}W_i\widehat{W}_ir_i\right\rangle_{R},
\end{equation}
where $W_i=R/N_i$, $\widehat{W}_i$ is the modular multiplicative inverse of $W_i$ modulo $N_i$, i.e.,
$W_i\widehat{W}_i\equiv 1 \!\!\mod N_i$,
(or equivalently, $\widehat{W}_i$ is some integer satisfying
\begin{equation}
W_i\widehat{W}_i+N_iQ_i=1
\end{equation}
for some integer $Q_i$), if $N_i\neq1$, else $\widehat{W}_i=0$, and $N_1,N_2,\cdots,N_L$ are taken to be any $L$ pairwise coprime positive integers such that $R=N_1N_2\cdots N_L$ and $N_i$ divides $M_i$ for each $1\leq i\leq L$.
\end{proposition}

It is worth noting that when the moduli $M_1, M_2,\cdots,M_L$ are pairwise coprime, we can take $N_i=M_i$ for $1\leq i\leq L$, and then Proposition \ref{crt_integer} reduces to the CRT for integers with respect to pairwise coprime moduli.

We next extend the CRT for integers to the integer vector reconstruction problem. We call it the MD-CRT for integer vectors. The non-commutativity of matrix multiplication prevents many results for integers from being clearly established for integer vectors and integer matrices. For this reason, it is necessary to explicitly derive the MD-CRT for integer vectors in this paper. Before presenting the main results, we first give the following lemma, which will be used in the sequel.

\begin{lemma}\label{lem1}
Given integer matrices $\textbf{M}_1,\textbf{M}_2,\cdots,\textbf{M}_L$, if $\textbf{B}$ is an lcrm of $\textbf{M}_1,\textbf{M}_2,\cdots,\textbf{M}_{L-1}$, and $\textbf{R}$ is an lcrm of $\textbf{B}$ and $\textbf{M}_L$, then $\textbf{R}$ is an lcrm of $\textbf{M}_1,\textbf{M}_2,\cdots,\textbf{M}_L$. In addition, a similar statement holds when lcrm above is replaced with lclm.
\end{lemma}
\begin{proof}
See Appendix \ref{ap1}.
\end{proof}

We then have the following result.
\begin{theorem}[MD-CRT for integer vectors]\label{them1}
Given $L$ moduli $\textbf{M}_i$ for $1\leq i\leq L$, which are arbitrary nonsingular integer matrices, let $\textbf{R}$ be anyone of their lcrm's. For an integer vector $\textbf{m}\in\mathcal{N}(\textbf{R})$, we can uniquely reconstruct $\textbf{m}$ from its remainders $\textbf{r}_i=\langle \textbf{m}\rangle_{\textbf{M}_i}$.
\end{theorem}
\begin{proof}
Let $\textbf{G}_1$ and $\textbf{R}_1$ be a gcld and an lcrm of $\textbf{M}_1$ and $\textbf{M}_2$, respectively. Based on the Bezout's theorem in Proposition \ref{pr3}, we have, for some integer matrices $\textbf{P}_1$ and $\textbf{P}_2$, $\textbf{M}_1\textbf{P}_1+\textbf{M}_2\textbf{P}_2=\textbf{G}_1$, on both sides of which
we right-multiply $\textbf{G}_1^{-1}$ and obtain
\begin{equation}\label{them1eq2}
\textbf{M}_1\textbf{P}_1\textbf{G}_1^{-1}+\textbf{M}_2\textbf{P}_2\textbf{G}_1^{-1}=\textbf{I}.
\end{equation}
Let
\begin{equation}\label{them1eq3}
\textbf{m}_1=\textbf{M}_2\textbf{P}_2\textbf{G}_1^{-1}\textbf{r}_1+\textbf{M}_1\textbf{P}_1\textbf{G}_1^{-1}\textbf{r}_2.
\end{equation}
We next prove that $\textbf{m}_1$ given in (\ref{them1eq3}) is a solution of a system of congruences as follows:
\begin{equation}\label{them1eq4}
\left\{\begin{array}{ll}
\textbf{m}\equiv \textbf{r}_1 \!\!\mod \textbf{M}_1\\
\textbf{m}\equiv \textbf{r}_2 \!\!\mod \textbf{M}_2.\\
\end{array}\right.
\end{equation}
From (\ref{them1eq2}), we can rewrite (\ref{them1eq3}) as
\begin{equation} \label{them1eq5}
\begin{split}
\textbf{m}_1
 & = (\textbf{I}-\textbf{M}_1\textbf{P}_1\textbf{G}_1^{-1})\textbf{r}_1+\textbf{M}_1\textbf{P}_1\textbf{G}_1^{-1}\textbf{r}_2\\
 & = \textbf{r}_1+\textbf{M}_1\textbf{P}_1\textbf{G}_1^{-1}(\textbf{r}_2-\textbf{r}_1).
\end{split}
\end{equation}
One can see from (\ref{them1eq4}) that $\textbf{M}_1\textbf{n}_1-\textbf{M}_2\textbf{n}_2=\textbf{r}_2-\textbf{r}_1$ holds for some integer vectors $\textbf{n}_1$ and $\textbf{n}_2$, and thus we have $\textbf{G}_1^{-1}(\textbf{r}_2-\textbf{r}_1)$ $=\textbf{G}_1^{-1}\textbf{M}_1\textbf{n}_1-\textbf{G}_1^{-1}\textbf{M}_2\textbf{n}_2$.
Since $\textbf{G}_1$ is a gcld of $\textbf{M}_1$ and $\textbf{M}_2$, we know that $\textbf{G}_1^{-1}\textbf{M}_1$ and $\textbf{G}_1^{-1}\textbf{M}_2$ are integer matrices, and thus $\textbf{G}_1^{-1}(\textbf{r}_2-\textbf{r}_1)$ is an integer vector. Therefore,
$\textbf{m}_1$ given in (\ref{them1eq3}) is an integer vector, and
we have, from (\ref{them1eq5}), $\textbf{m}_1\equiv \textbf{r}_1 \!\!\mod \textbf{M}_1$.
Similarly, we can rewrite (\ref{them1eq3}) as $\textbf{m}_1=\textbf{r}_2+\textbf{M}_2\textbf{P}_2\textbf{G}_1^{-1}(\textbf{r}_1-\textbf{r}_2)$,
and $\textbf{m}_1$ given in (\ref{them1eq3}) satisfies $\textbf{m}_1\equiv \textbf{r}_2 \!\!\mod \textbf{M}_2$.
That is to say, $\textbf{m}_1$ given in (\ref{them1eq3}) is a solution of the system of congruences in (\ref{them1eq4}). Thus, we have $\textbf{m}-\textbf{m}_1\in\text{LAT}(\textbf{M}_1)$ and $\textbf{m}-\textbf{m}_1\in\text{LAT}(\textbf{M}_2)$. From Proposition \ref{pr6}, we have $\textbf{m}-\textbf{m}_1\in\text{LAT}(\textbf{R}_1)$, i.e., $\textbf{m}\equiv \textbf{m}_1 \!\!\mod \textbf{R}_1$.
Based on the cascade architecture of the congruences, we can accordingly calculate a solution $\textbf{m}_2$ of
\begin{equation}
\left\{\begin{array}{ll}
\textbf{m}\equiv \textbf{m}_1 \!\!\mod \textbf{R}_1\\
\textbf{m}\equiv \textbf{r}_3 \!\!\mod \textbf{M}_3.\\
\end{array}\right.
\end{equation}
Letting $\textbf{R}_2$ be an lcrm of $\textbf{R}_1$ and $\textbf{M}_3$, we have $\textbf{m}\equiv \textbf{m}_2 \!\!\mod \textbf{R}_2$.
Moreover, from Lemma \ref{lem1}, $\textbf{R}_2$ is an lcrm of $\textbf{M}_1,\textbf{M}_2$, and $\textbf{M}_3$.
Following the above procedure, we merge two congruences at a time until we calculate a solution $\textbf{m}_{L-1}$ of
\begin{equation}
\left\{\begin{array}{ll}
\textbf{m}\equiv \textbf{m}_{L-2} \!\!\mod \textbf{R}_{L-2}\\
\textbf{m}\equiv \textbf{r}_L \!\!\mod \textbf{M}_L,\\
\end{array}\right.
\end{equation}
where $\textbf{R}_{L-2}$ is an lcrm of $\textbf{M}_1,\textbf{M}_2,\cdots,\textbf{M}_{L-1}$.
Let $\textbf{R}_{L-1}$ be an lcrm of $\textbf{R}_{L-2}$ and $\textbf{M}_L$, and we have $\textbf{m}\equiv \textbf{m}_{L-1} \!\!\mod \textbf{R}_{L-1}$,
where we readily know from Lemma \ref{lem1} that $\textbf{R}_{L-1}$ is an lcrm of $\textbf{M}_1,\textbf{M}_2,\cdots,\textbf{M}_L$. Without loss of generality, we can let $\textbf{R}_{L-1}=\textbf{R}$. So, we can get $\textbf{m}\in\mathcal{N}(\textbf{R})$ as
\begin{equation}
\textbf{m}=\langle\textbf{m}_{L-1}\rangle_{\textbf{R}}.
\end{equation}
Finally, we prove the uniqueness of the solution for $\textbf{m}$ modulo $\textbf{R}$. Assume that there exists another solution $\textbf{m}'\in\mathcal{N}(\textbf{R})$ that satisfies $\textbf{r}_i=\langle\textbf{m}'\rangle_{\textbf{M}_i}$ for $1\leq i\leq L$. Let $\textbf{m}''=\textbf{m}-\textbf{m}'$. We know
$\textbf{m}''\equiv \bm{0} \!\!\mod \textbf{M}_i$ for $1\leq i\leq L$, that is,
\begin{equation}
\textbf{m}''\in\text{LAT}(\textbf{M}_1)\cap\text{LAT}(\textbf{M}_1)\cap\cdots\cap\text{LAT}(\textbf{M}_L)=\text{LAT}(\textbf{R}),
\end{equation}
where the last equality is valid due to Proposition \ref{pr6} and Lemma \ref{lem1}.
Hence, we have $\textbf{m}''\in\text{LAT}(\textbf{R})$, i.e.,
\begin{equation}\label{christ}
\textbf{m}''=\textbf{R}\textbf{k}\; \text{ for some integer vector }\textbf{k}.
\end{equation}
Since $\textbf{m},\textbf{m}'\in\mathcal{N}(\textbf{R})$ and $\textbf{m}''=\textbf{m}-\textbf{m}'$, we have
\begin{equation}
\textbf{m}''\in\{\textbf{n}\,|\, \textbf{n}=\textbf{R}\textbf{x}, \textbf{x}\in(-1,1)^{D}\text{ and }\textbf{n}\in\mathbb{Z}^{D}\},
\end{equation}
where $D$ is the length of $\textbf{m}''$. Since $\textbf{R}$ is nonsingular from the definition of lcrm, this implies $\textbf{k}=\bm{0}$ in (\ref{christ}), and thus $\textbf{m}''=\bm{0}$. The proof is completed.
\end{proof}

As it can be seen in the proof of Theorem \ref{them1}, a reconstruction algorithm for the MD-CRT for integer vectors is given as well, which solves the first two congruences, uses that result as the remainder with respect to an lcrm of the first two moduli, and combines this new congruence with the third congruence, and so on.
We assume that there exist $L$ pairwise commutative and coprime integer matrices, denoted by $\textbf{N}_1,\textbf{N}_2,\cdots,\textbf{N}_L$, such that $\textbf{R}=\textbf{N}_1\textbf{N}_2\cdots\textbf{N}_L\textbf{U}$ for some unimodular matrix $\textbf{U}$ and $\textbf{N}_i$ is a left divisor of $\textbf{M}_i$ for each $1\leq i\leq L$ in Theorem \ref{them1}.
Under this assumption, we can derive a simple reconstruction formula for the MD-CRT for integer vectors as follows.

\begin{lemma}\label{lem2}
Let $\textbf{N}_i$ for $1\leq i\leq L$ be $L$ nonsingular integer matrices, which are pairwise commutative and coprime, i.e., $\textbf{N}_i\textbf{N}_j=\textbf{N}_j\textbf{N}_i$, and $\textbf{N}_i$ and $\textbf{N}_j$ are coprime for each pair of $i$ and $j$, $1\leq i\neq j\leq L$. Then,
$\textbf{N}_{i_1}\textbf{N}_{i_2}\cdots\textbf{N}_{i_p}$ and $\textbf{N}_{j_1}\textbf{N}_{j_2}\cdots\textbf{N}_{j_q}$ are commutative and coprime for any subsets $\{i_1,i_2,\cdots,i_p\}\subset\{1,2,\cdots,L\}$ and $\{j_1,j_2,\cdots,j_q\}\subset \{1,2,\cdots,L\} \backslash\{i_1,i_2,\cdots,i_p\}$. Moreover, $\textbf{N}_{i_1}\textbf{N}_{i_2}\cdots\textbf{N}_{i_p}$ is an lcm of $\textbf{N}_{i_1}, \textbf{N}_{i_2}, \cdots, \textbf{N}_{i_p}$ for any subset $\{i_1,i_2,\cdots,i_p\}\subset\{1,2,\cdots,L\}$ with $p\geq2$.
\end{lemma}
\begin{proof}
See Appendix \ref{ap2}.
\end{proof}

\begin{corollary}\label{cor1}
Given $L$ moduli $\textbf{M}_i$ for $1\leq i\leq L$, which are arbitrary nonsingular integer matrices, let $\textbf{R}$ be anyone of their lcrm's. Let us assume that there exist $L$ pairwise commutative and coprime integer matrices, denoted by $\textbf{N}_1,\textbf{N}_2,\cdots,\textbf{N}_L$,
such that $\textbf{R}=\textbf{N}_1\textbf{N}_2\cdots\textbf{N}_L\textbf{U}$ for some unimodular matrix $\textbf{U}$ and $\textbf{N}_i$ is a left divisor of $\textbf{M}_i$ for each $1\leq i\leq L$.
For an integer vector $\textbf{m}\in\mathcal{N}(\textbf{R})$, we can uniquely reconstruct $\textbf{m}$ from its remainders $\textbf{r}_i=\langle\textbf{m}\rangle_{\textbf{M}_i}$ as
\begin{equation}\label{jiejie}
\textbf{m}=\left\langle\sum_{i=1}^{L}\textbf{W}_i\widehat{\textbf{W}}_i\textbf{r}_i\right\rangle_{\textbf{R}},
\end{equation}
where $\textbf{W}_i=\textbf{N}_1\cdots\textbf{N}_{i-1}\textbf{N}_{i+1}\cdots\textbf{N}_L$, and if $\textbf{N}_i$ is not unimodular, $\widehat{\textbf{W}}_i$ is some integer matrix satisfying
\begin{equation}\label{bbbbb}
\textbf{W}_i\widehat{\textbf{W}}_i+\textbf{N}_i\textbf{Q}_i=\textbf{I}
\end{equation}
for some integer matrix $\textbf{Q}_i$, and can be calculated by following the procedure (\ref{calcal})--(\ref{zhongdian}) in advance; otherwise $\widehat{\textbf{W}}_i=\textbf{0}$.
\end{corollary}
\begin{proof}
See Appendix \ref{ap3}.
\end{proof}

In what follows, let us see in detail some special cases of the MD-CRT for integer vectors, where the $L$ nonsingular moduli are given by
\begin{equation}\label{sspeal2}
\textbf{M}_i=\textbf{M}\bm{\Gamma}_i\; \text{ for }1\leq i\leq L,
\end{equation}
and $\textbf{M}$ and $\bm{\Gamma}_i$'s here are integer matrices. Clearly, the moduli given by (\ref{sspeal2}) are in general not commutative. For the specific moduli in (\ref{sspeal2}), we first prove the following lemma.
\begin{lemma}\label{caona}
For the moduli $\textbf{M}_i$'s in (\ref{sspeal2}), if $\textbf{A}$ is an lcrm of $\bm{\Gamma}_i$ for $1\leq i\leq L$, then
$\textbf{M}\textbf{A}$ is an lcrm of $\textbf{M}_i$ for $1\leq i\leq L$.
\end{lemma}
\begin{proof}
See Appendix \ref{ap4}.
\end{proof}

Then, we present the following results.

\begin{corollary}\label{cor_xiao}
Given $L$ nonsingular moduli $\textbf{M}_i=\textbf{M}\bm{\Gamma}_i$ for $1\leq i\leq L$,
where $\textbf{M},\bm{\Gamma}_1,\bm{\Gamma}_2,\cdots,\bm{\Gamma}_L$ are pairwise commutative and coprime integer matrices, let $\textbf{R}$ be anyone of their lcrm's,
i.e., $\textbf{R}=\textbf{M}\bm{\Gamma}_1\bm{\Gamma}_2\cdots\bm{\Gamma}_L\textbf{U}$ for any unimodular matrix $\textbf{U}$. For an integer vector $\textbf{m}\in \mathcal{N}(\textbf{R})$, we can uniquely reconstruct $\textbf{m}$ from its remainders $\textbf{r}_i=\langle\textbf{m}\rangle_{\textbf{M}_i}$ as in Corollary \ref{cor1}.
\end{corollary}
\begin{proof}
See Appendix \ref{ap5}.
\end{proof}

\begin{corollary}\label{cor2}
Given $L$ nonsingular moduli $\textbf{M}_i=\textbf{M}\bm{\Gamma}_i$ for $1\leq i\leq L$,
where $\textbf{M}$ is a unimodular matrix, and $\bm{\Gamma}_i$'s are pairwise commutative and coprime integer matrices,
let $\textbf{R}$ be anyone of their lcrm's, i.e., $\textbf{R}=\textbf{M}\bm{\Gamma}_1\bm{\Gamma}_2\cdots\bm{\Gamma}_L\textbf{U}$ for any unimodular matrix $\textbf{U}$. For an integer vector $\textbf{m}\in\mathcal{N}(\textbf{R})$, we can uniquely reconstruct $\textbf{m}$ from its remainders $\textbf{r}_i=\langle\textbf{m}\rangle_{\textbf{M}_i}$ as
\begin{equation}\label{crtsolution2}
\textbf{m}=\left\langle\sum_{i=1}^{L}\textbf{W}_i\widehat{\textbf{W}}_i\textbf{r}_i\right\rangle_{\textbf{R}},
\end{equation}
where $\textbf{W}_i=\textbf{M}\bm{\Gamma}_1\cdots\bm{\Gamma}_{i-1}\bm{\Gamma}_{i+1}\cdots\bm{\Gamma}_L$, and $\widehat{\textbf{W}}_i$ is some integer matrix satisfying
\begin{equation}\label{bbb2}
\textbf{W}_i\widehat{\textbf{W}}_i+\textbf{M}_i\textbf{Q}_i=\textbf{I}
\end{equation}
for some integer matrix $\textbf{Q}_i$ and can be calculated by following the procedure (\ref{calcal})--(\ref{zhongdian}) in advance.
\end{corollary}
\begin{proof}
See Appendix \ref{ap6}.
\end{proof}

Particularly, when $\textbf{M}$ is the identity matrix, i.e., $\textbf{M}=\textbf{I}$, Corollary \ref{cor2} reduces to the MD-CRT for integer vectors with respect to pairwise commutative and coprime moduli (which is simply denoted as the CC MD-CRT for integer vectors), as stated below, in comparison with the CRT for integers with respect to pairwise coprime moduli.

\begin{theorem}[CC MD-CRT for integer vectors]\label{them2}
Given $L$ nonsingular moduli $\textbf{M}_i$ for $1\leq i\leq L$, which are pairwise commutative and coprime integer matrices, let $\textbf{R}$ be anyone of their lcrm's, i.e., $\textbf{R}=\textbf{M}_1\textbf{M}_2\cdots\textbf{M}_L\textbf{U}$ for any unimodular matrix $\textbf{U}$. For an integer vector $\textbf{m}\in\mathcal{N}(\textbf{R})$,
we can uniquely reconstruct $\textbf{m}$ from its remainders $\textbf{r}_i=\langle \textbf{m}\rangle_{\textbf{M}_i}$ as in Corollary \ref{cor2} with $\textbf{M}=\textbf{I}$.
\end{theorem}

We next see another special case of the MD-CRT for integer vectors, where
the $L$ nonsingular moduli can be simultaneously diagonalized by using two common unimodular matrices, i.e.,
\begin{equation}\label{jiafafa}
\textbf{M}_i=\textbf{U}\bm{\Lambda}_i\textbf{V}\in\mathbb{Z}^{D\times D}\; \text{ for }1\leq i\leq L
\end{equation}
with $\bm{\Lambda}_i$'s being diagonal integer matrices, and $\textbf{U}$ and $\textbf{V}$ being unimodular matrices.
For each $1\leq i\leq L$, write $\bm{\Lambda}_i$ as $\bm{\Lambda}_i=\text{diag}(\Lambda_i(1,1),\Lambda_i(2,2),\cdots,\Lambda_i(D,D))$.
Let
\begin{equation}\label{shenmeshi}
\bm{\Lambda}=\text{diag}(\Lambda(1,1),\Lambda(2,2),\cdots,\Lambda(D,D)),
\end{equation}
and $\Lambda(j,j)$ be the lcm of $\Lambda_1(j,j),\Lambda_2(j,j),\cdots,\Lambda_L(j,j)$ for each $1\leq j\leq D$. It is readily verified that
$\bm{\Lambda}$ is an lcm of $\bm{\Lambda}_i$'s.

We next prove that $\bm{\Lambda}$ is also an lcrm of $\bm{\Lambda}_i\textbf{V}$ for $1\leq i\leq L$. Since $\bm{\Lambda}$ is an lcm of $\bm{\Lambda}_i$'s, we have $\bm{\Lambda}=\bm{\Lambda}_i\textbf{P}_i$ for some integer matrices $\textbf{P}_i$'s. Due to the unimodularity of $\textbf{V}$, we have $\bm{\Lambda}=\bm{\Lambda}_i\textbf{V}\textbf{V}^{-1}\textbf{P}_i$ and $\textbf{V}^{-1}\textbf{P}_i$ is an integer matrix for each $1\leq i\leq L$. So, $\bm{\Lambda}$ is a crm of $\bm{\Lambda}_i\textbf{V}$ for $1\leq i\leq L$. For any other crm $\textbf{Q}$ of $\bm{\Lambda}_i\textbf{V}$ for $1\leq i\leq L$, we have $\textbf{Q}=\bm{\Lambda}_i\textbf{V}\textbf{Q}_i$ for some integer matrices $\textbf{Q}_i$'s, which indicates that $\textbf{Q}$ is a crm of $\bm{\Lambda}_i$'s. Thus, $\textbf{Q}$ is a right multiple of $\bm{\Lambda}$, i.e., $\bm{\Lambda}$ is an lcrm of $\bm{\Lambda}_i\textbf{V}$ for $1\leq i\leq L$.
Furthermore, from Lemma \ref{caona}, we obtain that $\textbf{U}\bm{\Lambda}$ is an lcrm of $\textbf{M}_i$'s given by (\ref{jiafafa}).
Let $\textbf{R}$ be anyone of the lcrm's of $\textbf{M}_i$'s, i.e., $\textbf{R}=\textbf{U}\bm{\Lambda}\textbf{B}$ for any unimodular matrix $\textbf{B}$.
For an integer vector $\textbf{m}\in\mathcal{N}(\textbf{R})$ and its remainders $\textbf{r}_i=\langle\textbf{m}\rangle_{\textbf{M}_i}$,
we have
\begin{equation}\label{liugn}
\textbf{m}= \textbf{U}\bm{\Lambda}_i\textbf{V}\textbf{n}_i+\textbf{r}_i \text{ and then }
\textbf{U}^{-1}\textbf{m}= \bm{\Lambda}_i\textbf{V}\textbf{n}_i+\textbf{U}^{-1}\textbf{r}_i,
\end{equation}
for $1\leq i\leq L$.
Due to the unimodularity of $\textbf{U}$, $\textbf{U}^{-1}\textbf{m}$ and $\textbf{U}^{-1}\textbf{r}_i$'s are all integer vectors.
Hence, we can view (\ref{liugn}) as a system of congruences with respect to the moduli $\bm{\Lambda}_i$'s, i.e.,
\begin{equation}\label{caolll}
\textbf{U}^{-1}\textbf{m}\equiv \textbf{U}^{-1}\textbf{r}_i \!\!\mod \bm{\Lambda}_i\; \text{ for }1\leq i\leq L,
\end{equation}
and then calculate the remainders $\bm{\zeta}_i\in\mathcal{N}(\bm{\Lambda}_i)$ of $\textbf{U}^{-1}\textbf{r}_i$ modulo $\bm{\Lambda}_i$,
i.e., $\textbf{U}^{-1}\textbf{r}_i\equiv \bm{\zeta}_i \!\!\mod \bm{\Lambda}_i$, for $1\leq i\leq L$. From (\ref{caolll}), we get
\begin{equation}
\textbf{U}^{-1}\textbf{m}\equiv \bm{\zeta}_i \!\!\mod \bm{\Lambda}_i\; \text{ for }1\leq i\leq L.
\end{equation}
Since $\textbf{U}$ is unimodular and $\textbf{m}\in\mathcal{N}(\textbf{U}\bm{\Lambda}\textbf{B})$ for any unimodular matrix $\textbf{B}$, we have $\textbf{U}^{-1}\textbf{m}\in\mathcal{N}(\bm{\Lambda}\textbf{B})$. Furthermore, as $\bm{\Lambda}_i$'s are diagonal integer matrices, it is always ready to find
$L$ pairwise commutative and coprime integer matrices (i.e., coprime diagonal integer matrices), denoted by $\textbf{N}_1,\textbf{N}_2,\cdots,\textbf{N}_L$, such that $\bm{\Lambda}=\textbf{N}_1\textbf{N}_2\cdots\textbf{N}_L$ and $\textbf{N}_i$ is a left divisor of $\bm{\Lambda}_i$ for each $1\leq i\leq L$. Therefore, from Corollary \ref{cor1}, we can uniquely reconstruct such $\textbf{m}$. When $\textbf{m}$ is restricted to $\textbf{m}\in\mathcal{N}(\textbf{U}\bm{\Lambda})$, i.e., the unimodular matrix $\textbf{B}$ is taken to be the identity matrix, the reconstruction of $\textbf{m}$ is equivalent to that via the $D$ independent conventional CRT for integers as follows. Let $\textbf{a}=\textbf{U}^{-1}\textbf{m}\in\mathbb{Z}^{D}$. Because of $\textbf{m}\in\mathcal{N}(\textbf{U}\bm{\Lambda})$, we obtain $\textbf{a}\in\mathcal{N}(\bm{\Lambda})$. That is to say, every element $a(j)$ of $\textbf{a}$ satisfies $a(j)\in\mathcal{N}(\Lambda(j,j))$ (i.e., $0\leq a(j)<\Lambda(j,j)$) for $1\leq j\leq D$. Therefore, via the CRT for integers, we can uniquely reconstruct $a(j)$ for each $1\leq j\leq D$ in the following system of congruences:
\begin{equation}
a(j)\equiv \zeta_i(j) \!\!\mod |\Lambda_i(j,j)|\; \text{ for }1\leq i\leq L.
\end{equation}

Based on the above analysis, we have the following result.

\begin{corollary}\label{shuju}
Let $L$ nonsingular moduli $\textbf{M}_i$ for $1\leq i\leq L$ be given by (\ref{jiafafa}), and $\textbf{R}$ be anyone of their lcrm's, i.e., $\textbf{R}=\textbf{U}\bm{\Lambda}\textbf{B}$ for any unimodular matrix $\textbf{B}$, where $\bm{\Lambda}$ is given by (\ref{shenmeshi}). For an integer vector $\textbf{m}\in\mathcal{N}(\textbf{R})$, we can uniquely reconstruct $\textbf{m}$ from its remainders $\textbf{r}_i=\langle \textbf{m}\rangle_{\textbf{M}_i}$ as in Corollary \ref{cor1}. Interestingly, when $\textbf{B}$ is the identity matrix, i.e., $\textbf{R}=\textbf{U}\bm{\Lambda}$, the reconstruction of $\textbf{m}\in\mathcal{N}(\textbf{R})$ is equivalent to that via the $D$ independent conventional CRT for integers.
\end{corollary}

In particular, when the $D\times D$ nonsingular moduli $\textbf{M}_i$'s can be simultaneously diagonalized as
\begin{equation}\label{sspeal}
\textbf{M}_i=\textbf{U}\bm{\Lambda}_i\textbf{U}^{-1}\; \text{ for }1\leq i\leq L,
\end{equation}
where $\textbf{U}$ is a $D\times D$ unimodular matrix, and $\bm{\Lambda}_i$'s are diagonal integer matrices that are pairwise coprime, it is readily verified that the moduli are pairwise commutative and coprime.
Note that $\bm{\Lambda}_i$ and $\bm{\Lambda}_j$ are coprime if and only if their corresponding diagonal elements $\Lambda_i(k,k)$ and $\Lambda_j(k,k)$ are coprime for each $1\leq k\leq D$. For this case, as a direct consequence of Theorem \ref{them2} or Corollary \ref{shuju}, we obtain the following result, which has been presented in \cite{jiawenxian,PPV1}.

\begin{corollary}[\!\!\cite{PPV1}]\label{cor4}
Let $L$ nonsingular moduli $\textbf{M}_i$ for $1\leq i\leq L$ be given by (\ref{sspeal}), and $\textbf{R}$ be anyone of their lcrm's, i.e., $\textbf{R}=\textbf{U}\bm{\Lambda}_1\bm{\Lambda}_2\cdots\bm{\Lambda}_L\textbf{B}$ for any unimodular matrix $\textbf{B}$. For an integer vector $\textbf{m}\in\mathcal{N}(\textbf{R})$, we can uniquely reconstruct $\textbf{m}$ from its remainders $\textbf{r}_i=\langle\textbf{m}\rangle_{\textbf{M}_i}$ as in Theorem \ref{them2}.
\end{corollary}

It is worth pointing out that the results of the MD-CRT for integer vectors in this section are closely related to the already established results on the abstract CRT for rings \cite{crt_integer}. In the context of the non-commutative ring $\mathbb{Z}^{D\times D}$ of integer matrices, let $\mathcal{M}_i=\textbf{M}_i\mathbb{Z}^{D\times D}$ for $1\leq i\leq L$ be right ideals in $\mathbb{Z}^{D\times D}$, where $\textbf{M}_i$'s are pairwise left coprime. Let $\mathbb{Z}^{D\times D}/\textbf{M}_i\mathbb{Z}^{D\times D}$, called the quotient ring of $\mathbb{Z}^{D\times D}$ by $\mathcal{M}_i$, be defined as the set of all cosets of $\mathcal{M}_i$ (i.e., $\mathbb{Z}^{D\times D}/\textbf{M}_i\mathbb{Z}^{D\times D}=\{\textbf{R}+\mathcal{M}_i\,|\,\, \textbf{R}\in\mathbb{Z}^{D\times D}\}$). Based on the Bezout's theorem in Proposition \ref{pr3}, there exists a ring isomorphism
\begin{equation}\label{isomap}
\mathbb{Z}^{D\times D}/\cap_{i}\mathcal{M}_i\cong\mathbb{Z}^{D\times D}/\mathcal{M}_1\oplus \mathbb{Z}^{D\times D}/\mathcal{M}_2\oplus\cdots\oplus\mathbb{Z}^{D\times D}/\mathcal{M}_L,
\end{equation}
where $\oplus$ stands for the direct product of rings. Given $\textbf{A}, \textbf{B}\in\mathbb{Z}^{D\times D}$, $\textbf{A}$ is congruent to $\textbf{B}$ modulo $\textbf{M}_i$ if and only if $\textbf{A}-\textbf{B}\in\mathcal{M}_i$. The elements in $\mathbb{Z}^{D\times D}/\mathcal{M}_i$ can be taken as remainder classes of $\mathbb{Z}^{D\times D}$ modulo $\textbf{M}_i$. We can rephrase the isomorphism in (\ref{isomap}) by stating that
the system of congruences modulo $\textbf{M}_i$'s can be solved uniquely. This result can be correspondingly generalized to the case with arbitrary nonsingular integer matrices $\textbf{M}_i$'s.
As the modulo operation on a matrix is carried out independently along every column of the matrix, the MD-CRT for integer vectors in this paper can be regarded as a special case of the abstract CRT in the specific algebraic settings of $\mathbb{Z}^{D\times D}$.

We conclude this section by showing an example to explain how to implement the MD-CRT for integer vectors step by step. Since this example involves a family of $2\times2$ integer circulant matrices, let us first introduce some existing results for integer circulant matrices.
An integer matrix is said to be circulant if each row can be obtained from the preceding row by a right circular shift, e.g., a $2\times 2$ integer circulant matrix $\textbf{P}=\left(
                  \begin{array}{cc}
                    p & q \\
                    q & p \\
                  \end{array}
                \right)$.
It is obvious that integer circulant matrices are commutative with one another. It has been proved in \cite{PPV5,PPV6} that any two $2\times 2$ integer circulant matrices
$\textbf{P}_1=\left(
                  \begin{array}{cc}
                    p_1 & q_1 \\
                    q_1 & p_1 \\
                  \end{array}
                \right)$ and
                $\textbf{P}_2=\left(
                  \begin{array}{cc}
                    p_2 & q_2 \\
                    q_2 & p_2 \\
                  \end{array}
                \right)$
are coprime, if and only if $p_1+q_1$ is coprime with $p_2+q_2$, and $p_1-q_1$ is coprime with $p_2-q_2$. A trivial subclass of $2\times 2$ integer circulant matrices with all equal elements is excluded from consideration in \cite{PPV5,PPV6} and this paper.
We then prove that a $2\times2$ integer circulant matrix $\textbf{P}=\left(
                  \begin{array}{cc}
                    p & q \\
                    q & p \\
                  \end{array}
                \right)$
with $q\neq0$ cannot be diagonalized as in (\ref{sspeal}).

\begin{lemma}\label{wwwwww}
A $2\times2$ integer circulant matrix $\textbf{P}=\left(
                  \begin{array}{cc}
                    p & q \\
                    q & p \\
                  \end{array}
                \right)$ with $q\neq0$ cannot be diagonalized as $\textbf{P}=\textbf{U}\bm{\Lambda}\textbf{U}^{-1}$, where $\textbf{U}$ is a $2\times2$ unimodular matrix and $\bm{\Lambda}$ is a diagonal integer matrix.
\end{lemma}
\begin{proof}
See Appendix \ref{ap7}.
\end{proof}

\begin{example}
 Let $L=3$ moduli be
$\textbf{M}_i=\textbf{M}\bm{\Gamma}_i$ for $1\leq i\leq3$, where $\bm{\Gamma}_i$'s are pairwise coprime integer circulant matrices, i.e., $\bm{\Gamma}_1=\left(
                  \begin{array}{cc}
                    4 & -1 \\
                    -1 & 4 \\
                  \end{array}
                \right)$, $\bm{\Gamma}_2=\left(
                  \begin{array}{cc}
                    7 & 4 \\
                    4 & 7 \\
                  \end{array}
                \right)$, and $\bm{\Gamma}_3=\left(
                  \begin{array}{cc}
                    -2 & 6 \\
                    6 & -2 \\
                  \end{array}
                \right)$.
We then consider the following two cases that are not covered by \cite{PPV1} or Corollary \ref{cor4} in this paper.
\begin{itemize}
  \item[$\romannumeral1)$] $\textbf{M}$ is commutative and coprime with each $\bm{\Gamma}_i$ for $1\leq i\leq3$. This case corresponds to the moduli given in Corollary \ref{cor_xiao}. Without loss of generality, we take $\textbf{M}=\left(
                  \begin{array}{cc}
                    4 & 3 \\
                    3 & 4 \\
                  \end{array}
                \right),$ which is also an integer circulant matrix.
One can see that $\textbf{M}_i$ for $1\leq i\leq3$ in this case are $2\times2$ integer circulant matrices and their nondiagonal elements are non-zero. Therefore, from Lemma \ref{wwwwww}, the moduli $\textbf{M}_i$'s cannot be diagonalized as in (\ref{sspeal}).
Let $\textbf{R}=\textbf{M}\bm{\Gamma}_1\bm{\Gamma}_2\bm{\Gamma}_3=\left(\begin{array}{cc}
                    402 & 522 \\
                    522 & 402 \\
                  \end{array}\right)$
and $\textbf{m}=\left(\begin{array}{cc}
                    402 & 522 \\
                    522 & 402 \\
                  \end{array}\right)\left(
     \begin{array}{c}
      1/6 \\
      1/2 \\
     \end{array}
     \right)=\left(
     \begin{array}{c}
      328 \\
      288 \\
     \end{array}
     \right)\in\mathcal{N}(\textbf{R})$. The remainders of $\textbf{m}$ modulo $\textbf{M}_i$ for $1\leq i\leq3$ are calculated from (\ref{comput2}), respectively, i.e., $\textbf{r}_1=\left(
                  \begin{array}{c}
                    14 \\
                    14 \\
                  \end{array}
                \right)$, $\textbf{r}_2=\left(
                  \begin{array}{c}
                    39 \\
                    38 \\
                  \end{array}
                \right)$, and $\textbf{r}_3=\left(
                  \begin{array}{c}
                    14 \\
                    14 \\
                  \end{array}
                \right)$.
Conversely, we can reconstruct $\textbf{m}$ from its remainders $\textbf{r}_i$ for $1\leq i\leq3$ via the MD-CRT for integer vectors in Corollary \ref{cor_xiao}. Let $\textbf{N}_1=\textbf{M}\bm{\Gamma}_1$, $\textbf{N}_2=\bm{\Gamma}_2$, and $\textbf{N}_3=\bm{\Gamma}_3$. Let $\textbf{W}_1=\textbf{N}_2\textbf{N}_3$, $\textbf{W}_2=\textbf{N}_1\textbf{N}_3$, and $\textbf{W}_3=\textbf{N}_1\textbf{N}_2$, and then by following the procedure (\ref{calcal})--(\ref{zhongdian}) to calculate the corresponding $\widehat{\textbf{W}}_i$ for each $1\leq i\leq 3$ in the Bezout's theorem such that (\ref{bbbbb}) holds, we get
$\widehat{\textbf{W}}_1=\left(\begin{array}{cc}
                    9 & -3 \\
                    23 & -7 \\
                  \end{array}\right)$, $\widehat{\textbf{W}}_2=\left(\begin{array}{cc}
                    11 & -4 \\
                    -3 & 1 \\
                  \end{array}\right)$, and $\widehat{\textbf{W}}_3=\left(\begin{array}{cc}
                    -7 & 8 \\
                    50 & -57 \\
                  \end{array}\right)$. Then, from the reconstruction formula in (\ref{jiejie}), we have
\begin{equation*}
\begin{split}
\textbf{m} & = \left\langle\sum_{i=1}^{3}\textbf{W}_i\widehat{\textbf{W}}_i\textbf{r}_i\right\rangle_{\textbf{R}} \\
 & = \left\langle \left(\begin{array}{c}
                    8456 \\
                    5096 \\
                  \end{array}\right)+\left(\begin{array}{c}
                    1196 \\
                    15436 \\
                  \end{array}\right)+\left(\begin{array}{c}
                    -8862 \\
                    -10542 \\
                  \end{array}\right) \right\rangle_{\textbf{R}}\\
& =\left\langle \left(\begin{array}{c}
                    790 \\
                    9990 \\
                  \end{array}\right)\right\rangle_{\textbf{R}}=\left(\begin{array}{c}
                    328 \\
                    288 \\
                  \end{array}\right).
\end{split}
\end{equation*}

  \item[$\romannumeral2)$] $\textbf{M}$ is an arbitrary nonsingular integer matrix, which is not commutative or coprime with $\bm{\Gamma}_i$'s. This case corresponds to the general moduli given in Theorem \ref{them1}. Without loss of generality, we take $\textbf{M}=\left(
                  \begin{array}{cc}
                    2 & 3 \\
                    4 & 5 \\
                  \end{array}
                \right)$.
Obviously, $\textbf{M}_i$'s are not pairwise commutative, and thus they cannot be diagonalized as in (\ref{sspeal}). Let
$\textbf{R}=\textbf{M}\bm{\Gamma}_1\bm{\Gamma}_2\bm{\Gamma}_3=\left(\begin{array}{cc}
                    390 & 270 \\
                    654 & 534 \\
                  \end{array}\right)$
and $\textbf{m}=\left(\begin{array}{cc}
                    390 & 270 \\
                    654 & 534 \\
                  \end{array}\right)\left(
     \begin{array}{c}
      1/2 \\
      1/3 \\
     \end{array}
     \right)=\left(
     \begin{array}{c}
      285 \\
      505 \\
     \end{array}
     \right)\in\mathcal{N}(\textbf{R})$. The remainders of $\textbf{m}$ modulo $\textbf{M}_i$ for $1\leq i\leq3$ are calculated from (\ref{comput2}), respectively, i.e., $\textbf{r}_1=\left(
                  \begin{array}{c}
                    5 \\
                    9 \\
                  \end{array}
                \right)$, $\textbf{r}_2=\left(
                  \begin{array}{c}
                    27 \\
                    49 \\
                  \end{array}
                \right)$, and $\textbf{r}_3=\left(
                  \begin{array}{c}
                    3 \\
                    7 \\
                  \end{array}
                \right)$.
Conversely, we can reconstruct $\textbf{m}$ from its remainders $\textbf{r}_i$ for $1\leq i\leq3$ via the MD-CRT for integer vectors in Theorem \ref{them1}. In this case, even though we do not have an explicit reconstruction formula as Case $\romannumeral1)$, we can reconstruct $\textbf{m}$ by following the algorithm exhibited in the proof of Theorem \ref{them1}. One can readily verify that $\textbf{M}$ and $\textbf{R}_1=\textbf{M}\bm{\Gamma}_1\bm{\Gamma}_2$ are a gcld and an lcrm of $\textbf{M}_1$ and $\textbf{M}_2$, respectively. Based on the Bezout's theorem in Proposition \ref{pr3}, we follow the procedure (\ref{calcal})--(\ref{zhongdian}) to get $\textbf{P}_1=\left(\begin{array}{cc}
                    3 & 11 \\
                    1 & 4 \\
                  \end{array}\right)$ and $\textbf{P}_2=\left(\begin{array}{cc}
                    -2 & -8 \\
                    1 & 4 \\
                  \end{array}\right)$ such that $\textbf{M}_1\textbf{P}_1+\textbf{M}_2\textbf{P}_2=\textbf{M}$. So, from (\ref{them1eq3}), we obtain
$\textbf{m}_1=\textbf{M}_2\textbf{P}_2\textbf{M}^{-1}\textbf{r}_1+\textbf{M}_1\textbf{P}_1\textbf{M}^{-1}\textbf{r}_2=\left(
                  \begin{array}{c}
                    510 \\
                    994 \\
                  \end{array}
                \right)$,
which satisfies
\begin{equation*}
\left\{\begin{array}{ll}
\textbf{m}_1\equiv \textbf{r}_1 \!\!\mod \textbf{M}_1\\
\textbf{m}_1\equiv \textbf{r}_2 \!\!\mod \textbf{M}_2.\\
\end{array}\right.
\end{equation*}
We then calculate the remainder $\bm{\nu}_1$ of $\textbf{m}_1$ modulo $\textbf{R}_1$, i.e.,
$\bm{\nu}_1=\langle\textbf{m}_1\rangle_{\textbf{R}_1}=\left(
                  \begin{array}{c}
                    30 \\
                    52 \\
                  \end{array}
                \right)$.
Following the above procedure, we calculate a solution of a system of congruences:
\begin{equation*}
\left\{\begin{array}{ll}
\textbf{m}\equiv \bm{\nu}_1 \!\!\mod \textbf{R}_1\\
\textbf{m}\equiv \textbf{r}_3 \!\!\mod \textbf{M}_3.\\
\end{array}\right.
\end{equation*}
It is also readily verified that $\textbf{M}$ and $\textbf{R}_2=\textbf{R}=\textbf{M}\bm{\Gamma}_1\bm{\Gamma}_2\bm{\Gamma}_3$ are a gcld and an lcrm of $\textbf{R}_1$ and $\textbf{M}_3$, respectively.
Based on the Bezout's theorem in Proposition \ref{pr3}, we follow the procedure (\ref{calcal})--(\ref{zhongdian}) to get $\textbf{Q}_1=\left(\begin{array}{cc}
                    8 & -21 \\
                    -7 & 18 \\
                  \end{array}\right)$ and $\textbf{Q}_2=\left(\begin{array}{cc}
                    10 & -24 \\
                    -18 & 49 \\
                  \end{array}\right)$ such that $\textbf{R}_1\textbf{Q}_1+\textbf{M}_3\textbf{Q}_2=\textbf{M}$. From (\ref{them1eq3}), we get
$\textbf{m}_2=\textbf{M}_3\textbf{Q}_2\textbf{M}^{-1}\bm{\nu}_1+\textbf{R}_1\textbf{Q}_1\textbf{M}^{-1}\textbf{r}_3=\left(
                  \begin{array}{c}
                    -375 \\
                    1429 \\
                  \end{array}
                \right)$.
Therefore, we get
\begin{equation*}
\textbf{m}=\langle\textbf{m}_2\rangle_{\textbf{R}_2}=\left(
                  \begin{array}{c}
                    285 \\
                    505 \\
                  \end{array}
                \right).
\end{equation*}

\end{itemize}
\end{example}

\section{Robust MD-CRT for Integer Vectors}\label{sec4}
In practice, signals of interest are usually subject to noise, and accordingly the detected remainders may be erroneous in many signal processing applications of the CRT.
To this end, the robust CRT for integers has been proposed in \cite{xiaxianggenxia,xiaoweili,wenjiewang} and further
dedicatedly studied in \cite{guangwuxu,yangyang,yangyang2,xiaoli1,xiaoli2}.
It basically says that even though every remainder has a small error, a large nonnegative integer can be robustly reconstructed in the sense that the reconstruction error is upper bounded by the remainder error bound.
In this section, motivated by the applications in MD signal processing, we want to extend the robust CRT for integers to the MD case, called the robust MD-CRT for integer vectors.
Before presenting that, we first review the robust CRT for integers in \cite{wenjiewang}, for comparison purposes.

\begin{proposition}[Robust CRT for integers \cite{wenjiewang}]\label{robustcrt_integer}
Let $L$ moduli be $M_i=M\Gamma_i$ for $1\leq i\leq L$, where $\Gamma_i$'s are pairwise coprime positive integers, and $M>1$ is an arbitrary positive integer. Let $R\triangleq M\Gamma_1\Gamma_2\cdots\Gamma_L$ be their lcm. For an integer $m\in\mathcal{N}(R)$ (i.e., $0\leq m<R$), let $r_i$'s be its remainders, i.e.,  $r_i=\langle m\rangle_{M_i}$ or
\begin{equation}
m=M_in_i+r_i\;\text{ for }1\leq i\leq L,
\end{equation}
where $n_i$'s are its folding integers. Let $\tilde{r}_i\triangleq r_i+\triangle r_i$, $1\leq i\leq L$, denote the erroneous remainders, where $\triangle r_i$'s are the remainder errors. From the erroneous remainders $\tilde{r}_i$'s, we can accurately determine the folding integers $n_i$'s, if and only if
\begin{equation}\label{chifan}
-\frac{M}{2}\leq \triangle r_i-\triangle r_1<\frac{M}{2}\;\text{ for }2\leq i\leq L.
\end{equation}
In addition, let $\tau$ be the remainder error bound, i.e., $|\triangle r_i|\leq\tau$ for $1\leq i\leq L$, and a simple sufficient condition for accurately determining the folding integers $n_i$'s is derived as
\begin{equation}
\tau<\frac{M}{4}.
\end{equation}
Once the folding integers $n_i$'s are accurately obtained, a robust reconstruction of $m$ can be calculated by
\begin{equation}
\hat{m}=\left[\frac{1}{L}\sum_{i=1}^{L}(M_in_i+\tilde{r}_i)\right],
\end{equation}
where $[\cdot]$ denotes the rounding operation. Obviously, the reconstruction error is upper bounded by $\tau$, i.e., $|\hat{m}-m|\leq\tau$.
\end{proposition}

In \cite{wenjiewang}, a closed-form algorithm for determining the folding integers $n_i$'s in Proposition \ref{robustcrt_integer} was proposed as well. For more information on the robust CRT for integers, we refer the reader to a thorough review in \cite{xiaoli}.

Motivated by Proposition \ref{robustcrt_integer} or \cite{wenjiewang}, we propose the robust MD-CRT for integer vectors through accurately determining the folding vectors in the rest of this section. Before that, let us first state two significant definitions related to lattices.

\begin{definition}[The shortest vector problem (SVP) on lattices]
For a lattice $\text{LAT}(\textbf{M})$ that is generated by a nonsingular matrix $\textbf{M}$, the minimum distance of $\text{LAT}(\textbf{M})$ is the smallest distance between any two lattice points:
\begin{equation}
\lambda_{\text{LAT}(\textbf{M})}=\min_{\substack{\textbf{w},\textbf{v}\in\text{LAT}(\textbf{M}), \\\textbf{w}\neq\textbf{v}}}\lVert\textbf{w}-\textbf{v}\rVert.
\end{equation}
It is obvious that lattices are closed under addition and subtraction operations. Therefore, the minimum distance of $\text{LAT}(\textbf{M})$ is equivalently defined as the length (magnitude) of the shortest non-zero lattice point:
\begin{equation}\label{zuiduanjuli}
\lambda_{\text{LAT}(\textbf{M})}=\min_{\textbf{v}\in\text{LAT}(\textbf{M})\backslash\{\textbf{0}\}}\lVert\textbf{v}\rVert.
\end{equation}
\end{definition}

\begin{definition}[The closest vector problem (CVP) on lattices]
For a lattice $\text{LAT}(\textbf{M})$ that is generated by a nonsingular matrix $\textbf{M}\in\mathbb{R}^{D\times D}$, given an arbitrary point $\textbf{w}\in\mathbb{R}^{D}$, we find a closest lattice point of $\text{LAT}(\textbf{M})$ to $\textbf{w}$ by
\begin{equation}
\text{dist}(\text{LAT}(\textbf{M}),\textbf{w})=\min_{\textbf{v}\in\text{LAT}(\textbf{M})}\lVert\textbf{v}-\textbf{w}\rVert.
\end{equation}
\end{definition}

The SVP and CVP are the two most important computational problems on lattices. The algorithms for solving these problems either exactly or approximately have been extensively studied \cite{latticproblem,latticproblem2}. Note that the distance above can be measured by any norm of vectors, e.g., the Euclidean norm $\lVert\textbf{v}\rVert_2=\sqrt{\sum_i\lvert v(i)\rvert^2}$,
the $\ell_1$ norm $\lVert\textbf{v}\rVert_1=\sum_i\lvert v(i)\rvert$, and the $\ell_\infty$ norm $\lVert\textbf{v}\rVert_\infty=\max_i\lvert v(i)\rvert$.

Let $L$ nonsingular moduli $\textbf{M}_i\in\mathbb{Z}^{D\times D}$ for $1\leq i\leq L$ be given by
\begin{equation}\label{rrrr}
\textbf{M}_i=\textbf{M}\bm{\Gamma}_i,
\end{equation}
where $\bm{\Gamma}_i\in\mathbb{Z}^{D\times D}$ for $1\leq i\leq L$ are pairwise commutative and coprime, and $\textbf{M}\in\mathbb{Z}^{D\times D}$. Define
\begin{equation}\label{rangerange}
\mathcal{A}_i=\left\{\textbf{m}\in\mathbb{Z}^{D}\,|\,\, \lfloor \textbf{M}_i^{-1}\textbf{m}\rfloor\in\mathcal{N}(\bm{\Gamma}_1\cdots\bm{\Gamma}_{i-1}\bm{\Gamma}_{i+1}\cdots\bm{\Gamma}_L\textbf{U}_i)\right\}
\end{equation}
for $1\leq i\leq L$, where $\textbf{U}_i\in\mathbb{Z}^{D\times D}$ is any unimodular matrix.
Let $\widetilde{\textbf{r}}_i=\textbf{r}_i+\triangle\textbf{r}_i\in\mathbb{Z}^{D}$ for $1\leq i\leq L$ be the erroneous remainders of an integer vector $\textbf{m}$ with respect to the moduli $\textbf{M}_i$'s, where $\textbf{r}_i$'s and $\triangle\textbf{r}_i$'s are the remainders and remainder errors, respectively.

Since $\textbf{r}_i$'s are the remainders of $\textbf{m}$ with respect to the moduli $\textbf{M}_i$'s in (\ref{rrrr}), we have
\begin{equation}\label{fangchengzu}
\left\{\begin{array}{ll}
\textbf{m}=\textbf{M}\bm{\Gamma}_1\textbf{n}_1+\textbf{r}_1\\
\textbf{m}=\textbf{M}\bm{\Gamma}_2\textbf{n}_2+\textbf{r}_2\\
\:\:\:\:\;\,\vdots\\
\textbf{m}=\textbf{M}\bm{\Gamma}_L\textbf{n}_L+\textbf{r}_L.
\end{array}\right.
\end{equation}
Without loss of generality, we assume that $\textbf{m}\in\mathcal{A}_1$. Therefore, we treat the first equation in (\ref{fangchengzu}) as a reference to be subtracted from the other $L-1$ equations, and we get
\begin{equation}\label{fangchengzu1}
\left\{\begin{array}{ll}
\textbf{M}\bm{\Gamma}_1\textbf{n}_1-\textbf{M}\bm{\Gamma}_2\textbf{n}_2=\textbf{r}_2-\textbf{r}_1\\
\textbf{M}\bm{\Gamma}_1\textbf{n}_1-\textbf{M}\bm{\Gamma}_3\textbf{n}_3=\textbf{r}_3-\textbf{r}_1\\
\:\:\:\:\quad\quad\quad\quad\quad\quad\;\;\,\vdots\\
\textbf{M}\bm{\Gamma}_1\textbf{n}_1-\textbf{M}\bm{\Gamma}_L\textbf{n}_L=\textbf{r}_L-\textbf{r}_1.\\
\end{array}\right.
\end{equation}
Left-multiplying $\textbf{M}^{-1}$ on both sides of all the equations in (\ref{fangchengzu1}), we obtain
\begin{equation}\label{fangchengzu2}
\left\{\begin{array}{ll}
\bm{\Gamma}_1\textbf{n}_1-\bm{\Gamma}_2\textbf{n}_2=\textbf{M}^{-1}(\textbf{r}_2-\textbf{r}_1)\\
\bm{\Gamma}_1\textbf{n}_1-\bm{\Gamma}_3\textbf{n}_3=\textbf{M}^{-1}(\textbf{r}_3-\textbf{r}_1)\\
\:\:\:\:\quad\quad\quad\quad\;\;\,\vdots\\
\bm{\Gamma}_1\textbf{n}_1-\bm{\Gamma}_L\textbf{n}_L=\textbf{M}^{-1}(\textbf{r}_L-\textbf{r}_1).\\
\end{array}\right.
\end{equation}
From (\ref{fangchengzu2}), we know that $\textbf{M}^{-1}(\textbf{r}_i-\textbf{r}_1)$ for $2\leq i\leq L$ are integer vectors, i.e.,
\begin{equation}
\textbf{r}_i-\textbf{r}_1\in \text{LAT}(\textbf{M}).
\end{equation}
We then perform the modulo-$\bm{\Gamma}_i$ operation on both sides of the corresponding $(i-1)$-th equation in (\ref{fangchengzu2}) for $2\leq i\leq L$ to get
\begin{equation}\label{fangchengzu3}
\left\{\begin{array}{ll}
\bm{\Gamma}_1\textbf{n}_1\equiv \textbf{0}  \!\!\mod \bm{\Gamma}_1\\
\bm{\Gamma}_1\textbf{n}_1\equiv \textbf{M}^{-1}(\textbf{r}_2-\textbf{r}_1)  \!\!\mod \bm{\Gamma}_2\\
\bm{\Gamma}_1\textbf{n}_1\equiv \textbf{M}^{-1}(\textbf{r}_3-\textbf{r}_1)  \!\!\mod \bm{\Gamma}_3\\
\:\:\,\quad\quad\vdots\\
\bm{\Gamma}_1\textbf{n}_1\equiv \textbf{M}^{-1}(\textbf{r}_L-\textbf{r}_1)  \!\!\mod \bm{\Gamma}_L,\\
\end{array}\right.
\end{equation}
where the first equation is always available.

Since we know the erroneous remainders $\widetilde{\textbf{r}}_i$'s rather than the remainders $\textbf{r}_i$'s, we estimate $\textbf{r}_i-\textbf{r}_1$ for each $2\leq i\leq L$ by using a closest lattice point $\textbf{v}_i$ of $\text{LAT}(\textbf{M})$ to $\widetilde{\textbf{r}}_i-\widetilde{\textbf{r}}_1$, i.e.,
\begin{equation}\label{minx}
\textbf{v}_i\triangleq\argmin_{\textbf{v}\in\text{LAT}(\textbf{M})}\lVert\textbf{v}-(\widetilde{\textbf{r}}_i-\widetilde{\textbf{r}}_1)\rVert.
\end{equation}
Let $\widetilde{\textbf{n}}_i$ for $1\leq i\leq L$ be a set of solutions of (\ref{fangchengzu2}) when $\textbf{r}_i-\textbf{r}_1$ for $2\leq i\leq L$ are replaced with $\textbf{v}_i$. In summary, we have the following Algorithm \ref{algo} for obtaining $\widetilde{\textbf{n}}_i$'s.

\begin{algorithm}[!t]
\caption{}

\vspace{2mm}
\begin{algorithmic}[1]\label{algo}
\STATE Calculate $\textbf{v}_i$ for $2\leq i\leq L$ in (\ref{minx}) from $\widetilde{\textbf{r}}_i$ for $1\leq i\leq L$.
\STATE Calculate the remainder $\bm{\zeta}_i$ of $\textbf{M}^{-1}\textbf{v}_i$ modulo $\bm{\Gamma}_i$ for each $2\leq i\leq L$, i.e.,
\begin{equation}
\textbf{M}^{-1}\textbf{v}_i\equiv \bm{\zeta}_i  \!\!\mod \bm{\Gamma}_i,
\end{equation}
where $\bm{\zeta}_i\in\mathcal{N}(\bm{\Gamma}_i)$.
\STATE Calculate $\bm{\chi}_1\triangleq\bm{\Gamma}_1\widetilde{\textbf{n}}_1\in\mathcal{N}(\bm{\Gamma}_1\bm{\Gamma}_2\cdots\bm{\Gamma}_L\textbf{U}_1)$ via the CC MD-CRT for integer vectors in Theorem \ref{them2} from the following system of congruences
\begin{equation}\label{fangchengzu4}
\left\{\begin{array}{ll}
\bm{\Gamma}_1\widetilde{\textbf{n}}_1\equiv \textbf{0}  \!\!\mod \bm{\Gamma}_1\\
\bm{\Gamma}_1\widetilde{\textbf{n}}_1\equiv \bm{\zeta}_2  \!\!\mod \bm{\Gamma}_2\\
\bm{\Gamma}_1\widetilde{\textbf{n}}_1\equiv \bm{\zeta}_3  \!\!\mod \bm{\Gamma}_3\\
\:\:\,\quad\quad\vdots\\
\bm{\Gamma}_1\widetilde{\textbf{n}}_1\equiv \bm{\zeta}_L  \!\!\mod \bm{\Gamma}_L.\\
\end{array}\right.
\end{equation}
\STATE Calculate $\widetilde{\textbf{n}}_1=\bm{\Gamma}_1^{-1}\bm{\chi}_1\in\mathcal{N}(\bm{\Gamma}_2\bm{\Gamma}_3\cdots\bm{\Gamma}_L\textbf{U}_1)$, and then
\begin{equation}\label{xiaxxxx}
\widetilde{\textbf{n}}_i=\bm{\Gamma}_i^{-1}(\bm{\chi}_1-\textbf{M}^{-1}\textbf{v}_i)\;\text{ for }2\leq i\leq L.
\end{equation}
\end{algorithmic}
\end{algorithm}

Based on Algorithm \ref{algo}, we have the following result.
\begin{theorem}[Robust MD-CRT for integer vectors--\uppercase\expandafter{\romannumeral1}]\label{them_rcrt}
Let $L$ nonsingular moduli be given by (\ref{rrrr}). For an integer vector $\textbf{m}\in\bigcup_{i=1}^{L}\mathcal{A}_i$
(assuming without loss of generality that $\textbf{m}\in\mathcal{A}_{1}$),
we can accurately determine the folding vectors $\textbf{n}_i$'s of $\textbf{m}$ from the erroneous remainders $\widetilde{\textbf{r}}_i$'s by Algorithm \ref{algo}, if and only if
\begin{equation}\label{condition}
\bm{\theta}_i=\textbf{0}\;\text{ for }2\leq i\leq L,
\end{equation}
where $\bm{\theta}_i$ is defined by
\begin{equation}
\bm{\theta}_i\triangleq\argmin_{\bm{\theta}\in\text{LAT}(\textbf{M})}\lVert\bm{\theta}-(\triangle\textbf{r}_i-\triangle\textbf{r}_1)\rVert.
\end{equation}
Besides, we present two simple sufficient conditions for accurately determining the folding vectors $\textbf{n}_i$'s as follows.
\begin{enumerate}
  \item \textit{Condition 1:} A sufficient condition is given by
  \begin{equation}\label{condition_suf1}
\lVert\triangle\textbf{r}_i-\triangle\textbf{r}_1\rVert<\frac{\lambda_{\text{LAT}(\textbf{M})}}{2}\;\text{ for }2\leq i\leq L.
\end{equation}
  \item \textit{Condition 2:} Let $\tau$ be the remainder error bound, i.e., $\lVert\triangle\textbf{r}_i\rVert\leq\tau$ for $1\leq i\leq L$, and then a much simpler sufficient condition is given by
\begin{equation}\label{condition_suf2}
\tau<\frac{\lambda_{\text{LAT}(\textbf{M})}}{4}.
\end{equation}
\end{enumerate}
Once the folding vectors $\textbf{n}_i$'s are accurately obtained, a robust reconstruction of $\textbf{m}$ can be calculated by $\widetilde{\textbf{m}}=\frac{1}{L}\sum_{i=1}^L(\textbf{M}_{i}\textbf{n}_{i}+\widetilde{\textbf{r}}_{i})$. Obviously, the reconstruction error is upper bounded by $\tau$, i.e.,
\begin{equation}
\lVert\widetilde{\textbf{m}}-\textbf{m}\rVert\leq\tau.
\end{equation}
\end{theorem}
\begin{proof}
We first prove the sufficiency. From (\ref{minx}), we have
\begin{equation}\label{minx2}
\textbf{v}_i\triangleq\argmin_{\textbf{v}\in\text{LAT}(\textbf{M})}\lVert\textbf{v}-(\textbf{r}_i-\textbf{r}_1)-(\triangle\textbf{r}_i-\triangle\textbf{r}_1)\rVert
\end{equation}
for $2\leq i\leq L$. As lattices are known to be closed under addition and subtraction operations, we take $\bm{\theta}=\textbf{v}-(\textbf{r}_i-\textbf{r}_1)\in\text{LAT}(\textbf{M})$, and then (\ref{minx2}) is equivalent to
\begin{equation}\label{minx3}
\bm{\theta}_i\triangleq\argmin_{\bm{\theta}\in\text{LAT}(\textbf{M})}\lVert\bm{\theta}-(\triangle\textbf{r}_i-\triangle\textbf{r}_1)\rVert
\end{equation}
for $2\leq i\leq L$. If the condition in (\ref{condition}), i.e., $\bm{\theta}_i=\textbf{0}\;\text{ for }2\leq i\leq L$, holds, we obtain $\textbf{v}_i=\textbf{r}_i-\textbf{r}_1$ for $2\leq i\leq L$. Then, from (\ref{fangchengzu3}) and (\ref{fangchengzu4}), $\bm{\Gamma}_1\textbf{n}_1$ and $\bm{\Gamma}_1\widetilde{\textbf{n}}_1$ have the same remainders $\bm{\zeta}_i$'s with respect to the moduli $\bm{\Gamma}_i$'s.
Due to $\textbf{m}\in\mathcal{A}_{1}$ and $\textbf{n}_1=\lfloor \textbf{M}_1^{-1}\textbf{m}\rfloor$, we obtain $\textbf{n}_1\in\mathcal{N}(\bm{\Gamma}_{2}\bm{\Gamma}_{3}\cdots\bm{\Gamma}_L\textbf{U}_1)$, where $\textbf{U}_1$ is any unimodular matrix, and thus $\bm{\Gamma}_1\textbf{n}_1\in\mathcal{N}(\bm{\Gamma}_1\bm{\Gamma}_{2}\cdots\bm{\Gamma}_L\textbf{U}_1)$. From (\ref{fangchengzu4}),
$\bm{\Gamma}_1\textbf{n}_1$ can be accurately determined by the CC MD-CRT for integer vectors in Theorem \ref{them2}, so can be $\textbf{n}_1$, i.e., $\widetilde{\textbf{n}}_1=\textbf{n}_1$. After obtaining $\textbf{n}_1$, we can accurately determine the other folding vectors $\textbf{n}_i$ for $2\leq i\leq L$ by substituting $\textbf{n}_1$ into (\ref{fangchengzu2}).
Therefore, we get $\widetilde{\textbf{n}}_i=\textbf{n}_i$ for $1\leq i\leq L$ in (\ref{xiaxxxx}).

We next prove the necessity. Assume that there exists at least one remainder error that does not satisfy (\ref{condition}). For example, the $k$-th remainder error $\triangle\textbf{r}_k$ with $2\leq k\leq L$ satisfies
\begin{equation}\label{xxxxxx}
\bm{\theta}_k\neq\textbf{0}.
\end{equation}
Therefore, $\textbf{v}_k$ in (\ref{minx}) does not equal $\textbf{r}_k-\textbf{r}_1$. We then have the following cases.

\textit{Case A}: There exists one $j$ with $2\leq j\leq L$ such that
\begin{equation}\label{cxax}
\bm{\theta}_j\notin\text{LAT}(\textbf{M}\bm{\Gamma}_j),
\end{equation}
i.e., $\bm{\theta}_j\neq\textbf{M}\bm{\Gamma}_j\textbf{k}$ for any integer vector $\textbf{k}$. We then prove that the remainders of $\textbf{M}^{-1}\textbf{v}_j$ and $\textbf{M}^{-1}(\textbf{r}_j-\textbf{r}_1)$ modulo $\bm{\Gamma}_j$ are different. Assume that $\textbf{M}^{-1}\textbf{v}_j$ and $\textbf{M}^{-1}(\textbf{r}_j-\textbf{r}_1)$ have the same remainder modulo $\bm{\Gamma}_j$, i.e.,
\begin{equation}\label{xiaonian}
\textbf{M}^{-1}\textbf{v}_j-\textbf{M}^{-1}(\textbf{r}_j-\textbf{r}_1)=\bm{\Gamma}_j\textbf{q}
\end{equation}
for some integer vector $\textbf{q}$. Left-multiplying $\textbf{M}$ on both sides of (\ref{xiaonian}), we get
$\textbf{v}_j-(\textbf{r}_j-\textbf{r}_1)=\textbf{M}\bm{\Gamma}_j\textbf{q}$,
i.e., $\bm{\theta}_j=\textbf{M}\bm{\Gamma}_j\textbf{q}$, which contradicts with (\ref{cxax}). Therefore, the remainders of $\textbf{M}^{-1}\textbf{v}_j$ and $\textbf{M}^{-1}(\textbf{r}_j-\textbf{r}_1)$ modulo $\bm{\Gamma}_j$ are different. As a consequence, $\bm{\chi}_1=\bm{\Gamma}_1\widetilde{\textbf{n}}_1$ obtained from the system of congruences in (\ref{fangchengzu4})
does not equal $\bm{\Gamma}_1\textbf{n}_1$ as in (\ref{fangchengzu3}), and hence $\widetilde{\textbf{n}}_1\neq\textbf{n}_1$.

\textit{Case B}: For each $2\leq i\leq L$, $\bm{\theta}_i\in\text{LAT}(\textbf{M}\bm{\Gamma}_i)$ but there exists at least one $j$ with $2\leq j\leq L$ such that $\bm{\theta}_j\neq\textbf{0}$; see, for example, that the $k$-th remainder error makes $\bm{\theta}_k\neq\textbf{0}$ according to (\ref{xxxxxx}), i.e., $\textbf{v}_k\neq\textbf{r}_k-\textbf{r}_1$.
Since $\textbf{v}_i=\bm{\theta}_i+(\textbf{r}_i-\textbf{r}_1)$ and $\bm{\theta}_i\in\text{LAT}(\textbf{M}\bm{\Gamma}_i)$ for $2\leq i\leq L$, we have
$\textbf{M}^{-1}\textbf{v}_i\equiv \textbf{M}^{-1}(\textbf{r}_i-\textbf{r}_1)  \!\!\mod \bm{\Gamma}_i$.
So, $\bm{\Gamma}_1\textbf{n}_1$ and $\bm{\Gamma}_1\widetilde{\textbf{n}}_1$ have the same remainders $\bm{\zeta}_i$'s with respect to the moduli $\bm{\Gamma}_i$'s, and $\textbf{n}_1$ can be accurately determined, i.e., $\widetilde{\textbf{n}}_1=\textbf{n}_1$. However, due to $\textbf{v}_k\neq\textbf{r}_k-\textbf{r}_1$, we have $\widetilde{\textbf{n}}_k\neq\textbf{n}_k$ from (\ref{xiaxxxx}).
This proves the necessity.

We finally prove the two simple sufficient conditions in (\ref{condition_suf1}) and (\ref{condition_suf2}) for accurately determining the folding vectors $\textbf{n}_i$'s, respectively.

$1)$ \textit{Condition 1}: Assume that there exists $\bm{\theta}_i\in\text{LAT}(\textbf{M})$ with $\bm{\theta}_i\neq\textbf{0}$ satisfying
\begin{equation}
\bm{\theta}_i=\argmin_{\bm{\theta}\in\text{LAT}(\textbf{M})}\lVert\bm{\theta}-(\triangle\textbf{r}_i-\triangle\textbf{r}_1)\rVert
\end{equation}
for each $2\leq i\leq L$. Then, we have
\begin{equation}
\begin{split}
\lVert\bm{\theta}_i\rVert
&=\lVert\bm{\theta}_i-(\triangle\textbf{r}_i-\triangle\textbf{r}_1)-(\textbf{0}-(\triangle\textbf{r}_i-\triangle\textbf{r}_1))\rVert\\
&\leq \lVert\bm{\theta}_i-(\triangle\textbf{r}_i-\triangle\textbf{r}_1)\rVert+\lVert\triangle\textbf{r}_i-\triangle\textbf{r}_1\rVert \\
& \leq 2\lVert\triangle\textbf{r}_i-\triangle\textbf{r}_1\rVert < \lambda_{\text{LAT}(\textbf{M})}\,,
\end{split}
\end{equation}
which contradicts with $\lVert\bm{\theta}_i\rVert\geq\lambda_{\text{LAT}(\textbf{M})}$. Thus, we obtain $\bm{\theta}_i=\textbf{0}$ for each $2\leq i\leq L$.

$2)$ \textit{Condition 2}: When $\lVert\triangle\textbf{r}_i\rVert\leq\tau$ for $1\leq i\leq L$, we have
\begin{equation}
\lVert\triangle\textbf{r}_l-\triangle\textbf{r}_1\rVert\leq\lVert\triangle\textbf{r}_l\rVert+\lVert\triangle\textbf{r}_1\rVert\leq2\tau<\frac{\lambda_{\text{LAT}(\textbf{M})}}{2}
\end{equation}
for $2\leq l\leq L$, which implies \textit{Condition 1}.

This completes the proof of the theorem.
\end{proof}

\begin{remark}
In the 1-dimensional case when $M$ is an arbitrary positive integer and $\Gamma_i$'s are pairwise coprime positive integers, we can readily verify that $\romannumeral1)$ $\mathcal{A}_1=\mathcal{A}_2=\cdots=\mathcal{A}_L=\bigcup_{i=1}^{L}\mathcal{A}_i=\mathcal{N}(M\Gamma_1\Gamma_2\cdots\Gamma_L)$,
and $\romannumeral2)$ the conditions in (\ref{condition}) and (\ref{condition_suf1}) imply each other, whereas $\romannumeral1)$ and $\romannumeral2)$ are generally not observed in the MD case. Therefore, in the 1-dimensional case, from Theorem \ref{them_rcrt},
it turns out that $\lvert\triangle r_i-\triangle r_1\rvert<\frac{M}{2}$ for $2\leq i\leq L$ is a necessary and sufficient condition for accurately determining the folding integers $n_i$ for $1\leq i\leq L$, which is very similar to the robust CRT for integers in Proposition \ref{robustcrt_integer}.
The only difference is that there is one more equality sign in the left side of (\ref{chifan}), which is due to the fact that the rounding operation instead of a norm on $\mathbb{R}$ is used in \cite{wenjiewang}.
\end{remark}

Interestingly, we observe that the result of the robust MD-CRT for integer vectors is dependent upon its reconstruction algorithm.
Different reconstruction algorithms might bring about different results of the robust MD-CRT for integer vectors.
In the following, we propose another reconstruction algorithm, by which a different result of the robust MD-CRT for integer vectors is derived.

By Proposition \ref{pr2}, we first calculate the Smith normal form of $\textbf{M}$ in (\ref{rrrr}) as
\begin{equation}\label{zhongyao}
\textbf{U}\textbf{M}\textbf{V}=\bm{\Lambda},
\end{equation}
where $\textbf{U}$ and $\textbf{V}$ are unimodular matrices, and $\bm{\Lambda}$ is a diagonal integer matrix. So, we have $\textbf{M}^{-1}=\textbf{V}\bm{\Lambda}^{-1}\textbf{U}$. From (\ref{fangchengzu2}), we get
\begin{equation}\label{xiaofang1}
\left\{\begin{array}{ll}
\bm{\Gamma}_1\textbf{n}_1-\bm{\Gamma}_2\textbf{n}_2=\textbf{V}\bm{\Lambda}^{-1}\textbf{U}(\textbf{r}_2-\textbf{r}_1)\\
\bm{\Gamma}_1\textbf{n}_1-\bm{\Gamma}_3\textbf{n}_3=\textbf{V}\bm{\Lambda}^{-1}\textbf{U}(\textbf{r}_3-\textbf{r}_1)\\
\:\:\:\:\quad\quad\quad\quad\;\;\,\vdots\\
\bm{\Gamma}_1\textbf{n}_1-\bm{\Gamma}_L\textbf{n}_L=\textbf{V}\bm{\Lambda}^{-1}\textbf{U}(\textbf{r}_L-\textbf{r}_1).\\
\end{array}\right.
\end{equation}
Left-multiplying $\textbf{V}^{-1}$ on both sides of all the equations in (\ref{xiaofang1}), we obtain
\begin{equation}\label{xiaofang2}
\left\{\begin{array}{ll}
\textbf{V}^{-1}\bm{\Gamma}_1\textbf{n}_1-\textbf{V}^{-1}\bm{\Gamma}_2\textbf{n}_2=\bm{\Lambda}^{-1}\textbf{U}(\textbf{r}_2-\textbf{r}_1)\\
\textbf{V}^{-1}\bm{\Gamma}_1\textbf{n}_1-\textbf{V}^{-1}\bm{\Gamma}_3\textbf{n}_3=\bm{\Lambda}^{-1}\textbf{U}(\textbf{r}_3-\textbf{r}_1)\\
\:\:\:\:\quad\quad\quad\quad\quad\quad\quad\;\;\,\,\vdots\\
\textbf{V}^{-1}\bm{\Gamma}_1\textbf{n}_1-\textbf{V}^{-1}\bm{\Gamma}_L\textbf{n}_L=\bm{\Lambda}^{-1}\textbf{U}(\textbf{r}_L-\textbf{r}_1).\\
\end{array}\right.
\end{equation}
We then perform the modulo-$\textbf{V}^{-1}\bm{\Gamma}_i$ operation on both sides of the corresponding $(i-1)$-th equation in (\ref{xiaofang2}) for $2\leq i\leq L$ to obtain
\begin{equation}\label{xiaofang3}
\left\{\begin{array}{ll}
\textbf{V}^{-1}\bm{\Gamma}_1\textbf{n}_1\equiv \textbf{0}  \!\!\mod \textbf{V}^{-1}\bm{\Gamma}_1\\
\textbf{V}^{-1}\bm{\Gamma}_1\textbf{n}_1\equiv \bm{\Lambda}^{-1}\textbf{U}(\textbf{r}_2-\textbf{r}_1)  \!\!\mod \textbf{V}^{-1}\bm{\Gamma}_2\\
\textbf{V}^{-1}\bm{\Gamma}_1\textbf{n}_1\equiv \bm{\Lambda}^{-1}\textbf{U}(\textbf{r}_3-\textbf{r}_1)  \!\!\mod \textbf{V}^{-1}\bm{\Gamma}_3\\
\:\:\,\quad\quad\quad\:\:\vdots\\
\textbf{V}^{-1}\bm{\Gamma}_1\textbf{n}_1\equiv \bm{\Lambda}^{-1}\textbf{U}(\textbf{r}_L-\textbf{r}_1)  \!\!\mod \textbf{V}^{-1}\bm{\Gamma}_L,\\
\end{array}\right.
\end{equation}
where the first equation is always available. From Lemma \ref{caona}, we know that $\textbf{V}^{-1}\bm{\Gamma}_1\bm{\Gamma}_2\cdots\bm{\Gamma}_L$ is an lcrm of the moduli $\textbf{V}^{-1}\bm{\Gamma}_i$'s in (\ref{xiaofang3}).
Because of $\textbf{m}\in\mathcal{A}_{1}$ and $\textbf{n}_1=\lfloor \textbf{M}_1^{-1}\textbf{m}\rfloor$, we obtain $\textbf{n}_1\in\mathcal{N}(\bm{\Gamma}_{2}\bm{\Gamma}_{3}\cdots\bm{\Gamma}_L\textbf{U}_1)$, where $\textbf{U}_1$ is any unimodular matrix, and thus $\textbf{V}^{-1}\bm{\Gamma}_1\textbf{n}_1\in\mathcal{N}(\textbf{V}^{-1}\bm{\Gamma}_1\bm{\Gamma}_{2}\cdots\bm{\Gamma}_L\textbf{U}_1)$. So, according to Corollary \ref{cor2},
$\textbf{V}^{-1}\bm{\Gamma}_1\textbf{n}_1$ can be accurately determined by the MD-CRT for integer vectors, so can be $\textbf{n}_1$.

We estimate $\textbf{U}(\textbf{r}_i-\textbf{r}_1)$ for each $2\leq i\leq L$ by using a closest lattice point $\textbf{p}_i$ of $\text{LAT}(\bm{\Lambda})$ to $\textbf{U}(\widetilde{\textbf{r}}_i-\widetilde{\textbf{r}}_1)$, i.e.,
\begin{equation}\label{minxminx}
\textbf{p}_i\triangleq\argmin_{\textbf{p}\in\text{LAT}(\bm{\Lambda})}\lVert\textbf{p}-\textbf{U}(\widetilde{\textbf{r}}_i-\widetilde{\textbf{r}}_1)\rVert.
\end{equation}
Due to $\textbf{U}(\textbf{r}_i-\textbf{r}_1)\in\text{LAT}(\bm{\Lambda})$ and the closeness of addition and subtraction operations on lattices, (\ref{minxminx}) is equivalent to
\begin{equation}\label{minxminx2}
\bm{\vartheta}_i\triangleq\argmin_{\bm{\vartheta}\in\text{LAT}(\bm{\Lambda})}\lVert\bm{\vartheta}-\textbf{U}(\triangle\textbf{r}_i-\triangle\textbf{r}_1)\rVert.
\end{equation}
Since the erroneous remainders $\widetilde{\textbf{r}}_i$'s are known rather than the remainders $\textbf{r}_i$'s, let $\widetilde{\textbf{n}}_i$ for $1\leq i\leq L$ be a set of solutions of (\ref{xiaofang2})
when $\textbf{U}(\textbf{r}_i-\textbf{r}_1)$ for $2\leq i\leq L$ are replaced with $\textbf{p}_i$. In summary, we have the following Algorithm \ref{algoalgo} for obtaining $\widetilde{\textbf{n}}_i$'s.

\begin{algorithm}[!t]
\caption{}

\vspace{2mm}
\begin{algorithmic}[1]\label{algoalgo}
\STATE Calculate $\textbf{p}_i$ for $2\leq i\leq L$ in (\ref{minxminx}) from $\widetilde{\textbf{r}}_i$ for $1\leq i\leq L$.
\STATE Calculate the remainder $\bm{\varpi}_i$ of $\bm{\Lambda}^{-1}\textbf{p}_i$ modulo $\textbf{V}^{-1}\bm{\Gamma}_i$ for each $2\leq i\leq L$, i.e.,
\begin{equation}
\bm{\Lambda}^{-1}\textbf{p}_i\equiv \bm{\varpi}_i  \!\!\mod \textbf{V}^{-1}\bm{\Gamma}_i,
\end{equation}
where $\bm{\varpi}_i\in\mathcal{N}(\textbf{V}^{-1}\bm{\Gamma}_i)$.
\STATE Calculate $\bm{\psi}_1\triangleq\textbf{V}^{-1}\bm{\Gamma}_1\widetilde{\textbf{n}}_1\in\mathcal{N}(\textbf{V}^{-1}\bm{\Gamma}_1\bm{\Gamma}_2\cdots\bm{\Gamma}_L\textbf{U}_1)$
via the MD-CRT for integer vectors in Corollary \ref{cor2} from the following system of congruences
\begin{equation}\label{xiaofang4}
\left\{\begin{array}{ll}
\textbf{V}^{-1}\bm{\Gamma}_1\widetilde{\textbf{n}}_1\equiv \textbf{0}  \!\!\mod \textbf{V}^{-1}\bm{\Gamma}_1\\
\textbf{V}^{-1}\bm{\Gamma}_1\widetilde{\textbf{n}}_1\equiv \bm{\varpi}_2  \!\!\mod \textbf{V}^{-1}\bm{\Gamma}_2\\
\textbf{V}^{-1}\bm{\Gamma}_1\widetilde{\textbf{n}}_1\equiv \bm{\varpi}_3  \!\!\mod \textbf{V}^{-1}\bm{\Gamma}_3\\
\:\:\,\quad\quad\quad\;\;\vdots\\
\textbf{V}^{-1}\bm{\Gamma}_1\widetilde{\textbf{n}}_1\equiv \bm{\varpi}_L  \!\!\mod \textbf{V}^{-1}\bm{\Gamma}_L.\\
\end{array}\right.
\end{equation}
\STATE Calculate $\widetilde{\textbf{n}}_1=\bm{\Gamma}_1^{-1}\textbf{V}\bm{\psi}_1\in\mathcal{N}(\bm{\Gamma}_2\bm{\Gamma}_3\cdots\bm{\Gamma}_L\textbf{U}_1)$, and then
\begin{equation}\label{xiaxxxx22}
\widetilde{\textbf{n}}_i=\bm{\Gamma}_i^{-1}\textbf{V}(\bm{\psi}_1-\bm{\Lambda}^{-1}\textbf{p}_i)\;\text{ for }2\leq i\leq L.
\end{equation}
\end{algorithmic}
\end{algorithm}

Based on Algorithm \ref{algoalgo}, we have the following result.
\begin{theorem}[Robust MD-CRT for integer vectors--\uppercase\expandafter{\romannumeral2}]\label{them_rcrt222}
Let $L$ nonsingular moduli be given by (\ref{rrrr}) and the Smith normal form of $\textbf{M}$ be given by (\ref{zhongyao}). For an integer vector
$\textbf{m}\in\bigcup_{i=1}^{L}\mathcal{A}_i$
(assuming without loss of generality that $\textbf{m}\in\mathcal{A}_{1}$),
we can accurately determine the folding vectors $\textbf{n}_i$'s of $\textbf{m}$ from the erroneous remainders $\widetilde{\textbf{r}}_i$'s by Algorithm \ref{algoalgo}, if and only if
\begin{equation}\label{condition2222}
\bm{\vartheta}_i=\textbf{0}\;\text{ for }2\leq i\leq L.
\end{equation}
Besides, we present two simple sufficient conditions for accurately determining the folding vectors $\textbf{n}_i$'s as follows.
\begin{enumerate}
  \item \textit{Condition 1}: A sufficient condition is given by
  \begin{equation}\label{condition_suf1_2222}
\lVert\textbf{U}(\triangle\textbf{r}_i-\triangle\textbf{r}_1)\rVert<\frac{\lambda_{\text{LAT}(\bm{\Lambda})}}{2}\;\text{ for }2\leq i\leq L.
\end{equation}
  \item \textit{Condition 2}: Let $\tau$ be the remainder error bound, i.e., $\lVert\triangle\textbf{r}_i\rVert\leq\tau$ for $1\leq i\leq L$, and then a much simpler sufficient condition is given by
\begin{equation}\label{condition_suf2_2222}
\tau<\frac{\lambda_{\text{LAT}(\bm{\Lambda})}}{4\lVert\textbf{U}\rVert_{\ast}},
\end{equation}
where $\lVert\textbf{U}\rVert_{\ast}$ stands for the subordinate matrix norm of $\textbf{U}$ based on the vector norm $\lVert\cdot\rVert$, i.e.,
$\lVert\textbf{U}\rVert_{\ast}=\underset{\lVert\textbf{x}\rVert=1}{\text{sup}}\{\lVert\textbf{U}\textbf{x}\rVert\}$.
\end{enumerate}
Once the folding vectors $\textbf{n}_i$'s are accurately obtained, a robust reconstruction of $\textbf{m}$ can be calculated by $\widetilde{\textbf{m}}=\frac{1}{L}\sum_{i=1}^L(\textbf{M}_{i}\textbf{n}_{i}+\widetilde{\textbf{r}}_{i})$. Obviously, the reconstruction error is upper bounded by $\tau$, i.e.,
\begin{equation}
\lVert\widetilde{\textbf{m}}-\textbf{m}\rVert\leq\tau.
\end{equation}
\end{theorem}

On the basis of the above analysis (\ref{zhongyao})--(\ref{minxminx2}), the proof of Theorem \ref{them_rcrt222} is similar to that of Theorem \ref{them_rcrt} and is thus omitted here. Let us take a simple example below to show
a difference between Theorem \ref{them_rcrt} and Theorem \ref{them_rcrt222} (between Algorithm \ref{algo} and Algorithm \ref{algoalgo}). Their difference is caused by
the non-equivalence of the conditions in (\ref{condition}) and (\ref{condition2222}).

\begin{example}
Let $\textbf{U}=\left(
                  \begin{array}{cc}
                    2 & 1 \\
                    1 & 1 \\
                  \end{array}
                \right)$ be a unimodular matrix, and $\textbf{M}$ in (\ref{rrrr}) be $\textbf{M}=\textbf{U}^{-1}\bm{\Lambda}\textbf{U}$, where $\bm{\Lambda}=\left(
                  \begin{array}{cc}
                    8 & 0 \\
                    0 & 8 \\
                  \end{array}
                \right)$. According to Proposition \ref{pr5}, we know
        \begin{equation*}
        \text{LAT}(\textbf{M})=\text{LAT}(\textbf{U}^{-1}\bm{\Lambda})=\text{LAT}\left(\left(
                  \begin{array}{cc}
                    8 & -8 \\
                    -8 & 16 \\
                  \end{array}
                \right)\right).
        \end{equation*}
Without loss of generality, we consider the first two remainder errors, i.e., $\triangle\textbf{r}_1$ and $\triangle\textbf{r}_2$. Let $\triangle\textbf{r}_2-\triangle\textbf{r}_1\triangleq\left(
                                                                                                                                                                                      \begin{array}{c}
                                                                                                                                                                                        \Delta_1 \\
                                                                                                                                                                                        \Delta_2 \\
                                                                                                                                                                                      \end{array}
                                                                                                                                                                                    \right).$
Then,
\begin{equation*}
\textbf{U}(\triangle\textbf{r}_2-\triangle\textbf{r}_1)=\left(
                                                          \begin{array}{c}
                                                            2\Delta_1+\Delta_2 \\
                                                            \Delta_1+\Delta_2 \\
                                                          \end{array}
                                                        \right).
\end{equation*}
In this example, we measure the distance by the Euclidean norm of vectors in $\mathbb{R}^2$. On one hand, let $\Delta_1=5$ and $\Delta_2=-8$. It is ready to verify that
$\bm{\vartheta}_2=\textbf{0}$, i.e., the condition in (\ref{condition2222}) holds for $i=2$. However, $\bm{\theta}_2\neq\textbf{0}$ in (\ref{condition}), since
$\triangle\textbf{r}_2-\triangle\textbf{r}_1$ is much closer to
a non-zero lattice point, e.g., $\left(
                                                                                                                                                                                  \begin{array}{c}
                                                                                                                                                                                    8 \\
                                                                                                                                                                                    -8 \\
                                                                                                                                                                                  \end{array}
                                                                                                                                                                                \right)$,
of $\text{LAT}(\textbf{M})$ than to $\textbf{0}$.
On the other hand, let $\Delta_1=3$ and $\Delta_2=0$. It is ready to verify that $\bm{\theta}_2=\textbf{0}$, i.e., the condition in (\ref{condition}) holds for $i=2$.
However, $\bm{\vartheta}_2\neq\textbf{0}$ in (\ref{condition2222}), since
$\textbf{U}(\triangle\textbf{r}_2-\triangle\textbf{r}_1)$ is much closer to
a non-zero lattice point, e.g., $\left(
                                                                                                    \begin{array}{c}
                                                                                                      8 \\
                                                                                                      0 \\
                                                                                                    \end{array}
                                                                                                  \right)$,
of $\text{LAT}(\bm{\Lambda})$ than to $\textbf{0}$.
\end{example}

We shall make a remark that the MD-CRT and robust MD-CRT for integer vectors studied in this paper are different from the generalized CRT and robust generalized CRT for integers in \cite{xialao1,xialao2,xiaopingli1,huiyongliao,xiaoli5,hanshen1,hanshen2}.
In the generalized CRT and robust generalized CRT for integers, every modular is a positive integer and multiple large positive integers are reconstructed from their unordered remainder sets, where an unordered remainder set consists of the remainders of the multiple integers modulo one modular, but the correspondence between the multiple integers and their remainders in the remainder set is unknown.
However, in the MD-CRT and robust MD-CRT for integer vectors, every modular is a nonsingular integer matrix and an integer vector is reconstructed from its remainders, where a remainder is an integer vector.
In particular, when the moduli $\textbf{M}_i\in\mathbb{Z}^{D\times D}$ for $1\leq i\leq L$ are diagonal integer matrices with positive main diagonal elements, i.e., $\textbf{M}_i=\text{diag}(M_i(1,1),M_i(2,2),\cdots,M_i(D,D))$ with $M_i(j,j)>0$ for $1\leq j\leq D$ and $1\leq i\leq L$, let $\textbf{R}=\text{diag}(R(1,1),R(2,2),\cdots,R(D,D))$ be their lcrm, where $R(j,j)$ is the lcm of $M_1(j,j),M_2(j,j),\cdots,M_L(j,j)$ for each $1\leq j\leq D$.
Then, reconstruction of an integer vector $\textbf{m}=(m(1),m(2),\cdots,m(D))^{T}\in\mathcal{N}(\textbf{R})$ using the MD-CRT and robust MD-CRT for integer vectors is equivalent to reconstruction of all elements of the integer vector one by one using the CRT and robust CRT for integers, and is also equivalent to reconstruction of all elements of the integer vector using the generalized CRT and robust generalized CRT for integers with ordered remainder sets.

\section{Simulation Results}\label{sec5}
In this section, we first show numerical simulations to verify the robust MD-CRT for integer vectors. We moreover apply the robust MD-CRT for integer vectors to MD frequency estimation when a complex MD sinusoidal signal is undersampled by multiple sub-Nyquist sampling matrices. In all the experiments below, without loss of generality, we consider the robust MD-CRT for integer vectors in Theorem \ref{them_rcrt} (i.e., Algorithm \ref{algo}), where the integer vector or frequency to be estimated falls into the range of $\mathcal{A}_{1}$ with $\textbf{U}_1=\textbf{I}$ in (\ref{rangerange}), and the vector norm $\lVert\cdot\rVert$ is the Euclidean norm, i.e., $\lVert\cdot\rVert_2$. In the simulations, we solve the integer quadratic programming problems in (\ref{zuiduanjuli}) and (\ref{minx}) using enumeration \cite{enum} and MOSEK with CVX \cite{boyd}, respectively.

Let moduli be $\textbf{M}_i=\textbf{M}\bm{\Gamma}_i$ for $i=1,2$, where $\bm{\Gamma}_1=\left(
                                                                                         \begin{array}{cc}
                                                                                           1 & 3 \\
                                                                                           3 & 1 \\
                                                                                         \end{array}
                                                                                       \right)
$, $\bm{\Gamma}_2=\left(
                                                                                         \begin{array}{cc}
                                                                                           3 & 4 \\
                                                                                           4 & 3 \\
                                                                                         \end{array}
                                                                                       \right)
$, and two different $\textbf{M}$'s are considered, given by $\textbf{M}=\left(
                                                                     \begin{array}{cc}
                                                                       48 & 17 \\
                                                                       8 & 46 \\
                                                                     \end{array}
                                                                   \right)
$ and $\textbf{M}=2\left(
                                                                     \begin{array}{cc}
                                                                       48 & 17 \\
                                                                       8 & 46 \\
                                                                     \end{array}
                                                                   \right)=\left(
                                                                     \begin{array}{cc}
                                                                       96 & 34 \\
                                                                       16 & 92 \\
                                                                     \end{array}
                                                                   \right)$ for simplicity.
We can easily know from \cite{PPV5,PPV6} that the integer circulant matrices $\bm{\Gamma}_1$ and $\bm{\Gamma}_2$ are commutative and coprime. According to Theorem \ref{them_rcrt}, the two different $\textbf{M}$'s lead to two different remainder error bounds $\tau<48.66/4=12.17$ and $\tau<97.32/4=24.33$, respectively. With respect to each of $\textbf{M}$'s, we uniformly choose the unknown integer vector $\textbf{m}\in\mathcal{A}_{1}$ and two remainder errors
$\lVert\triangle\textbf{r}_1\rVert_2\leq\tau$ and $\lVert\triangle\textbf{r}_2\rVert_2\leq\tau$. In this simulation, we consider the remainder error bounds $\tau=0,2,4,6,\cdots,30$, and for each of them, $5000$ trials are run. We apply Algorithm \ref{algo} to get the estimate $\widetilde{\textbf{m}}$, and in Fig. \ref{fig1}, we illustrate the mean error $E(\lVert\textbf{m}-\widetilde{\textbf{m}}\rVert_2)$ in terms of various remainder error bounds. One can observe from Fig. \ref{fig1} that for both of the two cases with different $\textbf{M}$'s, all the reconstruction errors are much smaller than the remainder error bound $\tau$, until $\tau$ achieves the maximal possible bound. This coincides with the theoretical result in Theorem \ref{them_rcrt}.

\begin{figure}[H]
\centerline{\includegraphics[width=1.05\columnwidth,draft=false]{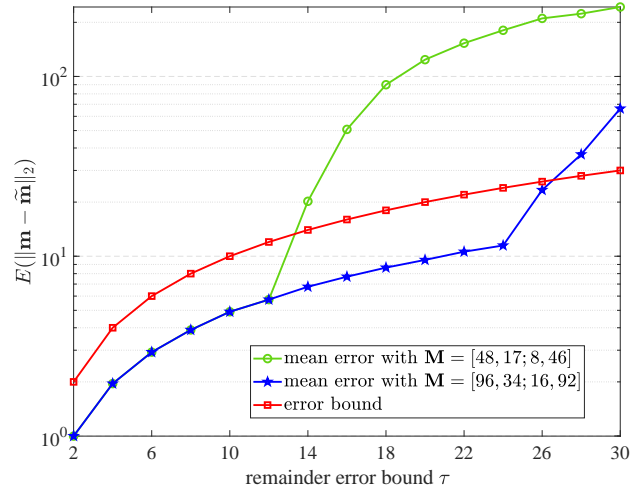}}
\caption{Mean error and theoretical error bound for the two cases with different $\textbf{M}$'s.}
 \label{fig1}
\end{figure}

Next, we formulate the application of the robust MD-CRT for integer vectors in MD sinusoidal frequency estimation as follows. Without loss of generality, suppose that $\textbf{f}=\textbf{m}\in\mathbb{Z}^D$ is an unknown integer MD frequency of interest in a complex MD sinusoidal signal $x(\textbf{t})$ and may be very high:
\begin{equation}\label{haihai}
x(\textbf{t})=a\exp(j2\pi \textbf{f}^T\textbf{t})+\omega(\textbf{t}),\; \textbf{t}\in\mathbb{R}^D,
\end{equation}
where $a$ is an unknown non-zero constant and $\omega(\textbf{t})$ is additive noise. Let $\textbf{M}_i^{-T}$ for $1\leq i\leq L$ be $L$ different sampling matrices that have sampling densities $|\text{det}(\textbf{M}_i)|$, respectively, where each $\textbf{M}_i$ is a $D\times D$ nonsingular integer matrix. For each sampling matrix $\textbf{M}_i^{-T}$, the undersampled sinusoidal signal is given by
\begin{equation}\label{haihai2}
x_i(\textbf{n})=a\exp(j2\pi \textbf{f}^T\textbf{M}_i^{-T}\textbf{n})+\omega_i(\textbf{n}),\; \textbf{n}\in\mathbb{Z}^D.
\end{equation}
We take the MD discrete Fourier transform (DFT) with respect to $\textbf{M}_i^{T}$ \cite{dft1} for the above $x_i(\textbf{n})$, $\textbf{n}\in\mathcal{N}(\textbf{M}_i^{T})$, and we have, for $\textbf{k}\in\mathcal{N}(\textbf{M}_i)$,
\begin{align}\label{xiouxi}
X_i(\textbf{k})&=a\sum_{\textbf{n}\in\mathcal{N}(\textbf{M}_i^T)}\exp(j2\pi \textbf{f}^T\textbf{M}_i^{-T}\textbf{n})\exp(-j2\pi\textbf{k}^T\textbf{M}_i^{-T}\textbf{n})+\Omega_i(\textbf{k})\nonumber\\
&=a\sum_{\textbf{n}\in\mathcal{N}(\textbf{M}_i^T)}\exp(-j2\pi (\textbf{k}-\textbf{f})^T\textbf{M}_i^{-T}\textbf{n})+\Omega_i(\textbf{k})\nonumber\\
&=a\sum_{\textbf{n}\in\mathcal{N}(\textbf{M}_i^T)}\exp(-j2\pi (\textbf{k}-\textbf{r}_i)^T\textbf{M}_i^{-T}\textbf{n})+\Omega_i(\textbf{k})\nonumber\\
&=a|\text{det}(\textbf{M}_i)|\delta(\textbf{k}-\textbf{r}_i)+\Omega_i(\textbf{k}),
\end{align}
where $\Omega_i(\textbf{k})$ is the MD DFT of $\omega_i(\textbf{n})$ with respect to $\textbf{M}_i^{T}$, $\textbf{r}_i$ is the remainder of $\textbf{f}$ modulo $\textbf{M}_i$, i.e., $\textbf{r}_i=\langle\textbf{f}\rangle_{\textbf{M}_i}$, and $\delta(\textbf{n})$ is the MD discrete delta function, i.e., $\delta(\textbf{n})=1$ if $\textbf{n}=\textbf{0}$ and $\delta(\textbf{n})=0$ otherwise. Note that the last equation in (\ref{xiouxi}) holds due to the unitarity of the MD DFT \cite{unitary}, i.e., for any nonsingular integer matrix $\textbf{M}\in\mathbb{Z}^{D\times D}$,
\begin{equation}
\sum_{\textbf{n}\in\mathcal{N}(\textbf{M})}\exp(-j2\pi\textbf{k}^T\textbf{M}^{-1}\textbf{n})=|\text{det}(\textbf{M})|\delta(\langle\textbf{k}\rangle_{\textbf{M}^T})\;\text{ for }\textbf{k}\in\mathbb{Z}^D.
\end{equation}
Therefore, we can detect the remainder $\textbf{r}_i$ of $\textbf{f}$ modulo $\textbf{M}_i$ (also called the aliased frequency) as a peak in magnitude of the MD DFT domain of $x_i(\textbf{n})$ in (\ref{xiouxi}), if the signal-to-noise ratio (SNR) is not too low. Nevertheless, when the SNR is not too high, the detected remainder $\widetilde{\textbf{r}}_i$ is most likely to be erroneous, i.e., $\widetilde{\textbf{r}}_i\triangleq\textbf{r}_i+\triangle\textbf{r}_i\in\mathcal{N}(\textbf{M}_i)$, where $\triangle\textbf{r}_i$ is the remainder error. Then, the robust MD-CRT for integer vectors provides an intuitive way to estimate $\textbf{f}$ from the erroneous remainders $\{\widetilde{\textbf{r}}_i\}_{i=1}^L$ modulo the corresponding moduli $\{\textbf{M}_i\}_{i=1}^L$. At this point, the sampling densities ($|\text{det}(\textbf{M}_i)|$ for $1\leq i\leq L$) of the multiple sub-Nyquist samplings  may be far less than the Nyquist sampling density that is defined by $|\text{det}(\textbf{R})|$, where $\textbf{R}$ is an lcrm of $\{\textbf{M}_i\}_{i=1}^L$.

\begin{figure}[H]
\centerline{\includegraphics[width=1.05\columnwidth,draft=false]{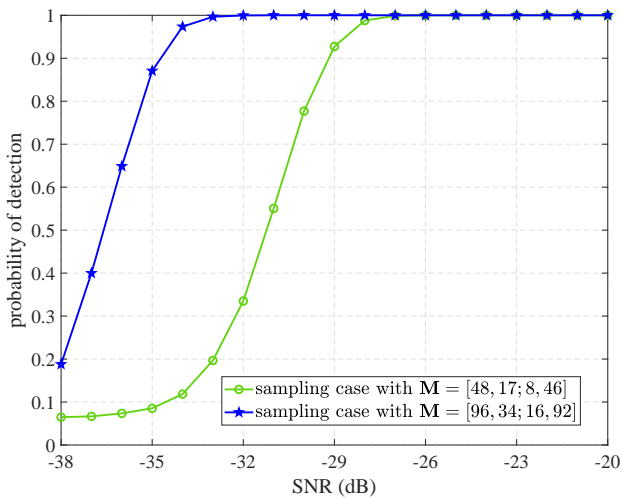}}
\caption{The probability of detection versus various SNR's for the two sampling cases with different $\textbf{M}$'s.}
 \label{fig2}
\end{figure}

We then illustrate the performance of the robust MD-CRT for integer vectors in the application of MD sinusoidal frequency estimation. In this simulation, we adopt the same $\textbf{M}_i=\textbf{M}\bm{\Gamma}_i$ for $i=1,2$ as in the first simulation (i.e., Fig. \ref{fig1}). Specifically, we undersample the sinusoidal signal with two sampling matrices $\textbf{M}_i^{-T}$ for $i=1,2$, followed by two MD DFT's with respect to $\textbf{M}_i^{T}$'s on the undersampled sinusoids, respectively, where we facilitate calculating the MD DFT's with respect to $\textbf{M}_i^{T}$'s by their equivalent separable MD DFT's \cite{jiawenxian}. We set $\textbf{f}=\left(
                                 \begin{array}{c}
                                   1645\\
                                   1373 \\
                                 \end{array}
                               \right)
$ and obviously this frequency belongs to $\mathcal{A}_{1}$ for both of the two sampling cases with different $\textbf{M}$'s. The additive noise in (\ref{haihai}) is complex white Gaussian noise, i.e., $\omega_i(\textbf{n})\sim\mathcal{CN}(0,2\sigma^2)$ in (\ref{haihai2}), and the SNR is defined as $\text{SNR}=10\log_{10}|a|^2/2\sigma^2\;\text{dB}$.
In Fig. \ref{fig2}, we present the probability of detection to illustrate the estimation performance of the robust MD-CRT for integer vectors in terms of various SNR's for the two sampling cases, where the estimated frequency $\textbf{f}$ is said to be correctly detected if its folding vectors are accurately determined, i.e., a robust estimate of $\textbf{f}$ is obtained, by Algorithm \ref{algo}. Fig. \ref{fig3} shows the mean relative error $E(\lVert\textbf{f}-\widetilde{\textbf{f}}\rVert_2/\lVert\textbf{f}\rVert_2)$ between the true $\textbf{f}$ and the reconstruction $\widetilde{\textbf{f}}$ verse SNR's for the two sampling cases. In these two figures, the SNR is increased from $-38$dB to $-20$dB and $5000$ trials are implemented for each SNR.
From Figs. \ref{fig2} and \ref{fig3}, the sampling case with $\textbf{M}=\left(
                                                                     \begin{array}{cc}
                                                                       96 & 34 \\
                                                                       16 & 92 \\
                                                                     \end{array}
                                                                   \right)$
achieves better performance (higher probability of detection and lower mean relative error) than the sampling case with $\textbf{M}=\left(
                                                                     \begin{array}{cc}
                                                                       48 & 17 \\
                                                                       8 & 46 \\
                                                                     \end{array}
                                                                   \right)$,
which is in accordance with the theoretical result in Theorem \ref{them_rcrt} that the former case has a larger robustness bound than the latter case, as mentioned before.

\begin{figure}[H]
\centerline{\includegraphics[width=1.05\columnwidth,draft=false]{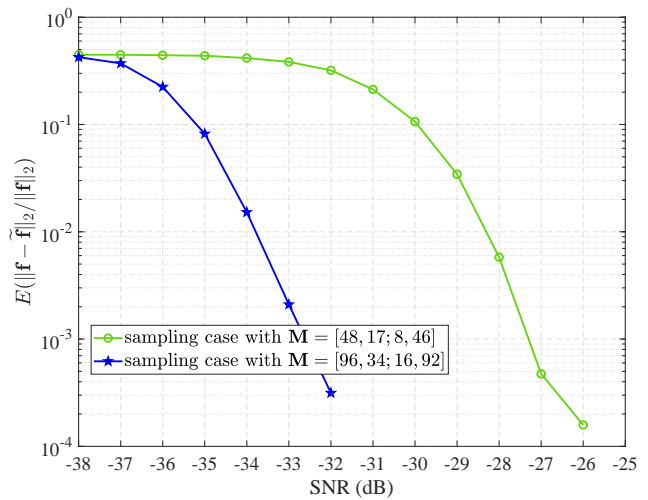}}
\caption{Mean relative error versus various SNR's for the two sampling cases with different $\textbf{M}$'s.}
 \label{fig3}
\end{figure}

As a final comment, general non-separable sampling matrices $\{\textbf{M}_i\}_{i=1}^L$ may lead to interesting MD signal processing properties as it has been pointed out earlier in the literature, for example, \cite{dft1,ppvbook}.

\section{Conclusion}\label{sec6}
In this paper, the CRT and robust CRT for integers are extended to the MD case, called the MD-CRT and robust MD-CRT for integer vectors, respectively, which are expected to have numerous applications in MD signal processing. Specifically, we first derive the MD-CRT for integer vectors with respect to a general set of moduli (namely a set of arbitrary nonsingular integer matrices), which allows to uniquely reconstruct an integer vector from its remainders, if the integer vector is in the fundamental parallelepiped of the lattice generated by an lcrm of all the moduli.
When the moduli are given in some special forms, we further present explicit reconstruction formulae. Furthermore,
we provide some results of the robust MD-CRT for integer vectors under the assumption that the remaining integer matrices of all the moduli left divided by their gcld are pairwise commutative and coprime.
Accordingly, we propose two different reconstruction algorithms, by which
two different conditions on the remainder error bound for the reconstruction robustness are separately obtained and proved to be related to a quarter of the minimum distance of the lattice generated by the gcld of all the moduli or the Smith normal form of the gcld. The robust MD-CRT for integer vectors with respect to a general set of moduli is still an open problem for future research.

\appendix

\subsection{Matrix Computation in the Bezout's theorem}\label{ap0}
Define $\textbf{S}=\left(
             \begin{array}{cc}
               \textbf{M} & \textbf{N} \\
             \end{array}\right)$,
which is a $D\times 2D$ integer matrix of rank $D$. From Proposition \ref{pr2}, the Smith normal form of $\textbf{S}$ is
\begin{equation}\label{calcal}
\textbf{U}\left(
             \begin{array}{cc}
               \textbf{M} & \textbf{N} \\
             \end{array}\right)\textbf{V}=\left(
             \begin{array}{cc}
               \bm{\Lambda} & \bm{0} \\
             \end{array}\right),
\end{equation}
where $\textbf{U}$ and $\textbf{V}$ are both unimodular matrices of sizes $D\times D$ and $2D\times 2D$, respectively, and $\bm{\Lambda}$ is a $D\times D$ diagonal integer matrix. Since $\textbf{U}^{-1}$ is also unimodular, we can write (\ref{calcal}) as
\begin{equation}\label{waas}
\left(\begin{array}{cc}
\textbf{M} & \textbf{N} \\
\end{array}\right)\textbf{V}=\left(
\begin{array}{cc}
\textbf{L} & \bm{0} \\
\end{array}\right),
\end{equation}
where $\textbf{L}=\textbf{U}^{-1}\bm{\Lambda}$ is a $D\times D$ integer matrix. Partitioning the $2D\times 2D$ unimodular matrix $\textbf{V}$ into $D\times D$ blocks, we have
\begin{equation}
\left(\begin{array}{cc}
\textbf{M} & \textbf{N} \\
\end{array}\right)\underbrace{\left(
                                \begin{array}{cc}
                                  \textbf{V}_{11} & \textbf{V}_{12} \\
                                  \textbf{V}_{21} & \textbf{V}_{22} \\
                                \end{array}
                              \right)
}_{\textbf{V}}=\left(
\begin{array}{cc}
\textbf{L} & \bm{0} \\
\end{array}\right).
\end{equation}
This implies
\begin{equation}\label{xxxxi}
\textbf{M}\textbf{V}_{11}+\textbf{N}\textbf{V}_{21}=\textbf{L}.
\end{equation}
By rewriting (\ref{waas}) as
\begin{equation}
\left(\begin{array}{cc}
\textbf{M} & \textbf{N} \\
\end{array}\right)=\left(
\begin{array}{cc}
\textbf{L} & \bm{0} \\
\end{array}\right)\underbrace{\left(
                    \begin{array}{cc}
                      \textbf{K}_{11} & \textbf{K}_{12} \\
                      \textbf{K}_{21} & \textbf{K}_{22}\\
                    \end{array}
                  \right)}_{\textbf{V}^{-1}},
\end{equation}
we have
$\textbf{M}=\textbf{L}\textbf{K}_{11}$ and $\textbf{N}=\textbf{L}\textbf{K}_{12}$,
where $\textbf{K}_{ij}$ for $1\leq i,j\leq2$ are all integer matrices due to the unimodularity of $\textbf{V}$.
Therefore, $\textbf{L}$ is a cld of $\textbf{M}$ and $\textbf{N}$. Then, we demonstrate that $\textbf{L}$ is actually a gcld of $\textbf{M}$ and $\textbf{N}$. For any other cld $\textbf{T}$ of $\textbf{M}$ and $\textbf{N}$, i.e., $\textbf{M}=\textbf{T}\textbf{A}$ and $\textbf{N}=\textbf{T}\textbf{B}$ for some integer matrices $\textbf{A}$ and $\textbf{B}$, we have, from (\ref{xxxxi}),
$\textbf{T}\left(\textbf{A}\textbf{V}_{11}+\textbf{B}\textbf{V}_{21}\right)=\textbf{L}$,
which means that $\textbf{T}$ is a left divisor of $\textbf{L}$. Therefore, $\textbf{L}$ is a gcld of $\textbf{M}$ and $\textbf{N}$, and is given by
\begin{equation}
\textbf{L}=\textbf{U}^{-1}\bm{\Lambda}.
\end{equation}
From (\ref{xxxxi}), the integer matrices $\textbf{P}$ and $\textbf{Q}$ in (\ref{btl}) are given by
\begin{equation}
\textbf{P}=\textbf{V}_{11} \text{ and } \textbf{Q}=\textbf{V}_{21}.
\end{equation}
In particular, if $\textbf{M}$ and $\textbf{N}$ are left coprime, their gcld $\textbf{L}$ must be unimodular. We right-multiply $\textbf{L}^{-1}$ on both sides of (\ref{xxxxi}), and can further get
\begin{equation}
\textbf{M}\textbf{V}_{11}\textbf{L}^{-1}+\textbf{N}\textbf{V}_{21}\textbf{L}^{-1}=\textbf{I}.
\end{equation}
This equation is called the Bezout's identity. In this case, $\textbf{I}$ is viewed as a gcld of $\textbf{M}$ and $\textbf{N}$, and the integer matrices
$\textbf{P}$ and $\textbf{Q}$ in (\ref{btl}) in the Bezout's identity are
\begin{equation}\label{zhongdian}
\textbf{P}=\textbf{V}_{11}\textbf{L}^{-1}=\textbf{V}_{11}\bm{\Lambda}^{-1}\textbf{U} \text{ and } \textbf{Q}=\textbf{V}_{21}\textbf{L}^{-1}=\textbf{V}_{21}\bm{\Lambda}^{-1}\textbf{U}.
\end{equation}

\subsection{Proof of Lemma \ref{lem1}}\label{ap1}
It is obvious that $\textbf{R}$ is a crm of $\textbf{M}_1,\textbf{M}_2,\cdots,\textbf{M}_L$. Then, for any other crm $\textbf{C}$ of $\textbf{M}_1,\textbf{M}_2,\cdots,\textbf{M}_L$, we have $\textbf{C}=\textbf{M}_L\textbf{Q}$ for some integer matrix $\textbf{Q}$. Moreover, since $\textbf{C}$ is a crm of $\textbf{M}_1,\textbf{M}_2,\cdots,\textbf{M}_{L-1}$, and $\textbf{B}$ is an lcrm of $\textbf{M}_1,\textbf{M}_2,\cdots,\textbf{M}_{L-1}$, we know that $\textbf{C}$ is a right multiple of $\textbf{B}$, i.e., $\textbf{C}=\textbf{B}\textbf{P}$
for some integer matrix $\textbf{P}$. Thus, $\textbf{C}$ is a crm of $\textbf{B}$ and $\textbf{M}_L$. Since $\textbf{R}$ is an lcrm of $\textbf{B}$ and $\textbf{M}_L$,  $\textbf{C}$ is known as a right multiple of $\textbf{R}$, i.e., $\textbf{C}=\textbf{R}\textbf{A}$ for some integer matrix $\textbf{A}$. Therefore, $\textbf{R}$ is an lcrm of $\textbf{M}_1,\textbf{M}_2,\cdots,\textbf{M}_L$. Similarly, we can prove that if $\textbf{B}$ is an lclm of $\textbf{M}_1,\textbf{M}_2,\cdots,\textbf{M}_{L-1}$, and $\textbf{R}$ is an lclm of $\textbf{B}$ and $\textbf{M}_L$, then $\textbf{R}$ is an lclm of $\textbf{M}_1,\textbf{M}_2,\cdots,\textbf{M}_L$.

\subsection{Proof of Lemma \ref{lem2}}\label{ap2}
As $\textbf{N}_1,\textbf{N}_2,\cdots,\textbf{N}_L$ are pairwise commutative, we immediately verify the commutativity of $\textbf{N}_{i_1}\textbf{N}_{i_2}\cdots\textbf{N}_{i_p}$ and $\textbf{N}_{j_1}\textbf{N}_{j_2}\cdots\textbf{N}_{j_q}$.
We next prove their coprimeness. For easy of presentation, we first look at a simple case when $L=3$. In this case, we need to prove without loss of generality that $\textbf{N}_1\textbf{N}_2$ and $\textbf{N}_3$ are coprime. Let $\textbf{D}$ be a gcrd of $\textbf{N}_1\textbf{N}_2$ and $\textbf{N}_3$. Since $\textbf{N}_1$ and $\textbf{N}_3$ are coprime, from the Bezout's theorem in Proposition \ref{pr3} we have, for some integer matrices $\textbf{P}$ and $\textbf{Q}$, $\textbf{P}\textbf{N}_1+\textbf{Q}\textbf{N}_3=\textbf{I}$, on both sides of which we
right-multiply $\textbf{N}_2\textbf{D}^{-1}$, and then commute $\textbf{N}_2$ and $\textbf{N}_3$ to get $\textbf{P}\textbf{N}_1\textbf{N}_2\textbf{D}^{-1}+\textbf{Q}\textbf{N}_2\textbf{N}_3\textbf{D}^{-1}=\textbf{N}_2\textbf{D}^{-1}$.
Since $\textbf{D}$ is a gcrd of $\textbf{N}_1\textbf{N}_2$ and $\textbf{N}_3$,
we know that $\textbf{N}_2\textbf{D}^{-1}$ is an integer matrix. That is to say,
$\textbf{D}$ is a right divisor of $\textbf{N}_2$. As stated above, $\textbf{D}$ is a right divisor of $\textbf{N}_3$, and $\textbf{N}_2$ and $\textbf{N}_3$ are coprime. So, $\textbf{D}$ must be a unimodular matrix. Thus, $\textbf{N}_1\textbf{N}_2$ and $\textbf{N}_3$ are right coprime (equivalently coprime from Proposition \ref{pr4} and their commutativity).
Accordingly, the above result can be readily generalized to the case when $L>3$, and therefore, $\textbf{N}_{i_1}\textbf{N}_{i_2}\cdots\textbf{N}_{i_p}$ and $\textbf{N}_{j_1}\textbf{N}_{j_2}\cdots\textbf{N}_{j_q}$ are coprime.

In addition, based on Proposition \ref{pr4} and the above result, we know that $\textbf{R}_2\triangleq\textbf{N}_{i_1}\textbf{N}_{i_2}$ is an lcm of $\textbf{N}_{i_1}$ and $\textbf{N}_{i_2}$, and $\textbf{R}_2$ is commutative and coprime with $\textbf{N}_{i_3}$. So,
$\textbf{R}_3\triangleq\textbf{R}_2\textbf{N}_{i_3}=\textbf{N}_{i_1}\textbf{N}_{i_2}\textbf{N}_{i_3}$ is an lcm of $\textbf{R}_2$ and $\textbf{N}_{i_3}$, and $\textbf{R}_3$ is commutative and coprime with $\textbf{N}_{i_4}$. Moreover, $\textbf{R}_3$ is an lcrm of $\textbf{N}_{i_1},\textbf{N}_{i_2}$, and $\textbf{N}_{i_3}$ from Lemma \ref{lem1}.
Continue this procedure until $\textbf{R}_p\triangleq\textbf{R}_{p-1}\textbf{N}_{i_p}=\textbf{N}_{i_1}\textbf{N}_{i_2}\cdots\textbf{N}_{i_p}$ is an lcm of $\textbf{R}_{p-1}$ and $\textbf{N}_{i_p}$. From Lemma \ref{lem1}, $\textbf{R}_p$ is an lcrm of $\textbf{N}_{i_1}, \textbf{N}_{i_2}, \cdots, \textbf{N}_{i_p}$, and similarly, we can deduce that $\textbf{R}_p$ is also an lclm of $\textbf{N}_{i_1}, \textbf{N}_{i_2}, \cdots, \textbf{N}_{i_p}$.

\subsection{Proof of Corollary \ref{cor1}}\label{ap3}
Since $\textbf{N}_i$ is a left divisor of $\textbf{M}_i$ for each $1\leq i\leq L$, there exists some integer matrix $\textbf{P}_i$ such that $\textbf{M}_i=\textbf{N}_i\textbf{P}_i$ for each $1\leq i\leq L$. So, from the remainders $\textbf{r}_i=\langle\textbf{m}\rangle_{\textbf{M}_i}$, we have
\begin{equation}\label{xiaox}
\textbf{m}=\textbf{N}_i\textbf{P}_i\textbf{n}_i+\textbf{r}_i\; \text{ for }1\leq i\leq L,
\end{equation}
where $\textbf{n}_i$'s are unknown folding vectors. Regarding (\ref{xiaox}) as a system of congruences with respect to the moduli $\textbf{N}_i$'s,
we get
\begin{equation}\label{pan}
\textbf{m}\equiv \bm{\xi}_i \!\!\mod \textbf{N}_i\; \text{ for }1\leq i\leq L,
\end{equation}
where $\bm{\xi}_i=\langle\textbf{r}_i\rangle_{\textbf{N}_i}$.
Since $\textbf{N}_1,\textbf{N}_2,\cdots,\textbf{N}_L$ are pairwise commutative and coprime, we know from Lemma \ref{lem2} that $\textbf{N}_1\textbf{N}_2\cdots\textbf{N}_L$ is their lcrm,
so is $\textbf{R}=\textbf{N}_1\textbf{N}_2\cdots\textbf{N}_L\textbf{U}$ for a unimodular matrix $\textbf{U}$.
Therefore, we obtain from Theorem \ref{them1} that $\textbf{m}\in\mathcal{N}(\textbf{R})$ can be uniquely reconstructed from its remainders $\bm{\xi}_i$'s or $\textbf{r}_i$'s. Next, we prove that $\textbf{m}$ in (\ref{jiejie}) is actually a solution of the system of congruences in (\ref{pan}). We express $\textbf{m}$ as
$\textbf{m}=\textbf{R}\textbf{n}+\sum_{i=1}^{L}\textbf{W}_i\widehat{\textbf{W}}_i\textbf{r}_i$
for some integer vector $\textbf{n}$.
Then, for each modulo-$\textbf{N}_j$ operation, we calculate
\begin{equation}
\begin{split}
\left\langle\textbf{m}\right\rangle_{\textbf{N}_j}
& = \left\langle\textbf{R}\textbf{n}+\textbf{W}_j\widehat{\textbf{W}}_j\textbf{r}_j+\sum_{i=1,i\neq j}^{L}\textbf{W}_i\widehat{\textbf{W}}_i\textbf{r}_i\right\rangle_{\textbf{N}_j} \\
& = \left\langle\textbf{W}_j\widehat{\textbf{W}}_j\textbf{r}_j\right\rangle_{\textbf{N}_j} = \left\langle\left(\textbf{I}-\textbf{N}_j\textbf{Q}_j\right)\textbf{r}_j\right\rangle_{\textbf{N}_j} \\
& = \langle\textbf{r}_j\rangle_{\textbf{N}_j}= \bm{\xi}_j\,,
\end{split}
\end{equation}
where the second equality is due to the commutativity of $\textbf{N}_i$'s, the third equality is obtained from (\ref{bbbbb}), and (\ref{bbbbb}) holds because $\textbf{N}_i$ is coprime with $\textbf{W}_i$ for each $1\leq i\leq L$ from Lemma \ref{lem2}. This completes the proof of the corollary.

\subsection{Proof of Lemma \ref{caona}}\label{ap4}
As $\textbf{A}$ is an lcrm of $\bm{\Gamma}_i$'s, $\textbf{M}\textbf{A}$ is a crm of $\textbf{M}_i$'s.
For any other crm $\textbf{C}$ of $\textbf{M}_i$'s, i.e., $\textbf{C}=\textbf{M}\bm{\Gamma}_i\textbf{P}_i$ for some integer matrices $\textbf{P}_i$'s, we have $\textbf{M}^{-1}\textbf{C}=\bm{\Gamma}_i\textbf{P}_i$, i.e.,
$\textbf{M}^{-1}\textbf{C}$ is a crm of $\bm{\Gamma}_i$'s. So, $\textbf{M}^{-1}\textbf{C}$ is a right multiple of $\textbf{A}$, i.e., $\textbf{M}^{-1}\textbf{C}=\textbf{A}\textbf{G}$ for some integer matrix $\textbf{G}$.
Hence, we have $\textbf{C}=\textbf{M}\textbf{A}\textbf{G}$. That is to say, $\textbf{M}\textbf{A}$ is an lcrm of $\textbf{M}_i$'s.

\subsection{Proof of Corollary \ref{cor_xiao}}\label{ap5}
Since $\bm{\Gamma}_i$'s are pairwise commutative and coprime, we know from Lemma \ref{lem2} that $\bm{\Gamma}_1\bm{\Gamma}_2\cdots\bm{\Gamma}_L$ is an lcm of $\bm{\Gamma}_i$'s.
Based on Lemma \ref{caona}, $\textbf{R}=\textbf{M}\bm{\Gamma}_1\bm{\Gamma}_2\cdots\bm{\Gamma}_L\textbf{U}$ for any unimodular matrix $\textbf{U}$ is an lcrm of $\textbf{M}_i$'s. Without loss of generality, we let $\textbf{N}_1=\textbf{M}\bm{\Gamma}_1$, $\textbf{N}_2=\bm{\Gamma}_2$, $\cdots$, $\textbf{N}_L=\bm{\Gamma}_L$. As $\textbf{M},\bm{\Gamma}_1,\bm{\Gamma}_2,\cdots,\bm{\Gamma}_L$ are pairwise commutative and coprime,
we obtain from Lemma \ref{lem2} that $\textbf{N}_i$'s are pairwise commutative and coprime. In addition, it is also readily seen that $\textbf{R}=\textbf{N}_1\textbf{N}_2\cdots\textbf{N}_L\textbf{U}$, and $\textbf{N}_i$ is a left divisor of $\textbf{M}_i$ for each $1\leq i\leq L$. Therefore, by Corollary \ref{cor1}, we can uniquely reconstruct $\textbf{m}\in \mathcal{N}(\textbf{R})$ from the moduli $\textbf{M}_i$'s and its remainders $\textbf{r}_i=\langle\textbf{m}\rangle_{\textbf{M}_i}$ by (\ref{jiejie}).

\subsection{Proof of Corollary \ref{cor2}}\label{ap6}
Since $\bm{\Gamma}_i$'s are pairwise commutative and coprime, $\bm{\Gamma}_1\cdots\bm{\Gamma}_{i-1}\bm{\Gamma}_{i+1}\cdots\bm{\Gamma}_L$ and $\bm{\Gamma}_i$ are known to be commutative and coprime from Lemma \ref{lem2}. We next prove that $\textbf{W}_i$ and $\textbf{M}_i$ are left coprime for each $1\leq i\leq L$.
Let $\textbf{D}$ be a gcld of $\textbf{W}_i$ and $\textbf{M}_i$. We then have $\textbf{W}_i=\textbf{M}\bm{\Gamma}_1\cdots\bm{\Gamma}_{i-1}\bm{\Gamma}_{i+1}\cdots\bm{\Gamma}_L=\textbf{D}\textbf{P}$ and $\textbf{M}_i=\textbf{M}\bm{\Gamma}_i=\textbf{D}\textbf{Q}$ for some integer matrices $\textbf{P}$ and $\textbf{Q}$.
Hence, we have $\bm{\Gamma}_1\cdots\bm{\Gamma}_{i-1}\bm{\Gamma}_{i+1}\cdots\bm{\Gamma}_L=\textbf{M}^{-1}\textbf{D}\textbf{P}$ and $\bm{\Gamma}_i=\textbf{M}^{-1}\textbf{D}\textbf{Q}$.
As $\textbf{M}$ is unimodular,
$\textbf{M}^{-1}\textbf{D}$ is an integer matrix and is a cld of $\bm{\Gamma}_1\cdots\bm{\Gamma}_{i-1}\bm{\Gamma}_{i+1}\cdots\bm{\Gamma}_L$ and $\bm{\Gamma}_i$. Since $\bm{\Gamma}_1\cdots\bm{\Gamma}_{i-1}\bm{\Gamma}_{i+1}\cdots\bm{\Gamma}_L$ and $\bm{\Gamma}_i$ are commutative and coprime, all of their cld's must be unimodular.
Therefore, $\textbf{M}^{-1}\textbf{D}$ is a unimodular matrix, so is $\textbf{D}$. That is to say, $\textbf{W}_i$ and $\textbf{M}_i$ are left coprime. Based on the Bezout's theorem in Proposition \ref{pr3}, there exist integer matrices, denoted by $\widehat{\textbf{W}}_i$ and $\textbf{Q}_i$, such that (\ref{bbb2}) holds for each $1\leq i\leq L$,
and we can calculate them by following the procedure (\ref{calcal})--(\ref{zhongdian}). In addition, from Lemma \ref{lem2} and Lemma \ref{caona}, we know that $\textbf{R}=\textbf{M}\bm{\Gamma}_1\bm{\Gamma}_2\cdots\bm{\Gamma}_L\textbf{U}$ for any unimodular matrix $\textbf{U}$ is an lcrm of the moduli $\textbf{M}_i$'s. The remaining proof is similar to the proof of Corollary \ref{cor1} and is omitted here.

\subsection{Proof of Lemma \ref{wwwwww}}\label{ap7}
Let $\bm{\alpha}_1=\left(
                  \begin{array}{c}
                    1 \\
                    1 \\
                  \end{array}
                \right)$ and $\bm{\alpha}_2=\left(
                  \begin{array}{c}
                    1 \\
                    -1 \\
                  \end{array}
                \right)$. It is readily checked that $\bm{\alpha}_1$ and $\bm{\alpha}_2$ are two eigenvectors of $\textbf{P}$ with the corresponding eigenvalues $p+q$ and $p-q$. Any integer vector in $\mathbb{R}^2$ can be represented by a linear combination of $\bm{\alpha}_1$ and $\bm{\alpha}_2$. Assume that $\textbf{P}$ can be diagonalized as $\textbf{P}=\textbf{U}\bm{\Lambda}\textbf{U}^{-1}$, where $\textbf{U}$ is a $2\times2$ unimodular matrix and $\bm{\Lambda}$ is a diagonal integer matrix. This means that $\textbf{U}$ is an eigenmatrix of $\textbf{P}$. Let $\textbf{u}$ be any column vector of $\textbf{U}$ and it can be represented by $\textbf{u}=a\bm{\alpha}_1+b\bm{\alpha}_2$ with $a,b\in\mathbb{R}$. Then, we get
$\textbf{P}\textbf{u}=p\textbf{u}+q(a\bm{\alpha}_1-b\bm{\alpha}_2)$.
Since $\textbf{U}$ is unimodular, it is nonsingular, and its column vectors cannot be the all-zero vectors. Since $\textbf{u}$ is a non-zero eigenvector of $\textbf{P}$ and $q\neq0$, we know that one and only one of $a$ and $b$ must be $0$.
Thus, $\textbf{U}$ has to be the form of $\left(
                  \begin{array}{cc}
                    a_1 & a_2 \\
                    a_1 & a_2 \\
                  \end{array}
                \right)$, $\left(
                  \begin{array}{cc}
                    b_1 & b_2 \\
                    -b_1 & -b_2 \\
                  \end{array}
                \right)$, $\left(
                  \begin{array}{cc}
                    a_1 & a_2 \\
                    a_1 & -a_2 \\
                  \end{array}
                \right)$, or $\left(
                  \begin{array}{cc}
                    b_1 & b_2 \\
                    -b_1 & b_2 \\
                  \end{array}
                \right)$.
Obviously, the determinants of $\left(
                  \begin{array}{cc}
                    a_1 & a_2 \\
                    a_1 & a_2 \\
                  \end{array}
                \right)$ and  $\left(
                  \begin{array}{cc}
                    b_1 & b_2 \\
                    -b_1 & -b_2 \\
                  \end{array}
                \right)$ are zero, which indicates that the first two forms of matrices are not possible to be unimodular. The determinants of $\left(
                  \begin{array}{cc}
                    a_1 & a_2 \\
                    a_1 & -a_2 \\
                  \end{array}
                \right)$ and $\left(
                  \begin{array}{cc}
                    b_1 & b_2 \\
                    -b_1 & b_2 \\
                  \end{array}
                \right)$ are $-2a_1a_2$ and $2b_1b_2$, respectively, which are not equal to $\pm1$ for any $a_1,a_2,b_1,b_2\in\mathbb{Z}$. This indicates that the last two forms of matrices are also not possible to be unimodular. Therefore, $\textbf{U}$ cannot be unimodular. That is to say, $\textbf{P}$ cannot be diagonalized as $\textbf{P}=\textbf{U}\bm{\Lambda}\textbf{U}^{-1}$, where $\textbf{U}$ is a $2\times2$ unimodular matrix and $\bm{\Lambda}$ is a diagonal integer matrix.


\begin{thebibliography}{1}

\bibitem{crt_integer}
D. S. Dummit and R. M. Foote,
{\em Abstract Algebra},
Hoboken: Wiley, 2004.

\bibitem{crt_integer1}
H. Krishna, B. Krishna, K.-Y. Lin, and J.-D. Sun,
{\em Computational Number Theory and Digital Signal Processing: Fast Algorithms and Error Control Techniques},
Boca Raton, FL: CRC, 1994.

\bibitem{crt_integer2}
C. Ding, D. Pei, and A. Salomaa,
{\em Chinese Remainder Theorem: Applications in Computing, Coding,
Cryptography}, Singapore: World Scientific, 1999.

\bibitem{conv1}
H. Krishna, K.-Y. Lin, and B. Krishna, ``Rings, fields, the Chinese remainder theorem and an extension--Part II: applications to digital signal processing,''
{\em IEEE Trans. Circuits Syst. II}, vol. 41, no. 10, pp. 656-668, 1994.

\bibitem{conv2}
R. C. Agarwal and C. S. Burrus, ``Number theoretic transforms to implement fast digital convolution,''
{\em Proc. IEEE}, vol. 63, no. 4, pp. 550-560, 1975.

\bibitem{fft1}
C.-F. Hsiao, Y. Chen, and C.-Y. Lee, ``A generalized mixed-radix algorithm for memory-based FFT processors,''
{\em IEEE Trans. Circuits Syst. II, Exp. Briefs}, vol. 57, no. 1, pp. 26-30, 2010.

\bibitem{fft2}
A. Christlieb, D. Lawlor, and Y. Wang, ``A multiscale sub-linear time Fourier algorithm for noisy data,''
{\em Appl. Comput. Harmon. Anal.}, vol. 40, no. 3, pp. 553-574, 2016.

\bibitem{coprime1}
P. P. Vaidyanathan and P. Pal, ``Sparse sensing with co-prime samplers and arrays,''
{\em IEEE Trans. Signal Process.}, vol. 59, no. 2, pp. 573-586, 2011.

\bibitem{coprime2}
S. Qin, Y. D. Zhang, and M. G. Amin, ``Generalized coprime array configurations for direction-of-arrival estimation,''
{\em IEEE Trans. Signal Process.}, vol. 63, no. 6, pp. 1377-1390, 2015.

\bibitem{conglin1}
C. H. Li, L. Gan, and C. Ling, ``Coprime sensing via Chinese remaindering over quadratic fields--Part I: Array designs,''
{\em IEEE Trans. Signal Process.}, vol. 67, no. 11, pp. 2898-2910, 2019.

\bibitem{conglin2}
C. H. Li, L. Gan, and C. Ling, ``Coprime sensing via Chinese remaindering over quadratic fields--Part II: Generalizations and applications,''
{\em IEEE Trans. Signal Process.}, vol. 67, no. 11, pp. 2911-2922, 2019.

\bibitem{xiaxianggenxia}
X.-G. Xia and G. Wang, ``Phase wrapping and a robust Chinese remainder theorem,''
{\em IEEE Signal Process. Lett.}, vol. 14, no. 4, pp. 247-250, 2007.

\bibitem{xiaoweili}
X. W. Li, H. Liang, and X.-G. Xia,
``A robust Chinese remainder theorem with its applications in frequency estimation from undersampled waveforms,''
{\em IEEE Trans. Signal Process.}, vol. 57, no. 11, pp. 4314-4322, 2009.

\bibitem{wenjiewang}
W. J. Wang and X.-G. Xia,
``A closed-form robust Chinese remainder theorem and its performance analysis,''
{\em IEEE Trans. Signal Process.}, vol. 58, no. 11, pp. 5655-5666, 2010.

\bibitem{gangli1}
G. Li, J. Xu, Y.-N. Peng, and X.-G. Xia,
``Location and imaging of moving targets using non-uniform linear antenna array,''
{\em IEEE Trans. Aerosp. Electron. Syst.}, vol. 43, no. 3, pp. 1214-1220, 2007.

\bibitem{yiminzhang}
Y. M. Zhang and M. Amin, ``MIMO radar exploiting narrowband frequency-hopping waveforms.'' In
{\em Proc. 16th European Signal Processing Conference (EUSIPCO 2008)}, Lausanne, Switzerland, 2008.

\bibitem{zihuiyuan}
Z. Yuan, Y. Deng, F. Li, R. Wang, G. Liu, and X. Han,
``Multichannel InSAR DEM reconstruction through improved closed-form robust Chinese remainder theorem,''
{\em IEEE Geosci. Remote Sens. Lett.}, vol. 10, no. 6, pp. 1314-1318, 2013.

\bibitem{radeee}
B. Silva and G. Fraidenraich,
``Performance analysis of the classic and robust Chinese remainder theorems in pulsed Doppler radars,''
{\em IEEE Trans. Signal Process.}, vol. 66, no. 18, pp. 4898-4903, 2018.

\bibitem{jiaxu1}
J. Xu, Z.-Z. Huang, Z. Wang, L. Xiao, X.-G. Xia, and T. Long,
``Radial velocity retrieval for multichannel SAR moving targets with time-space Doppler deambiguity,''
{\em IEEE Trans. Geosci. Remote Sens.}, vol. 56, no. 1, pp. 35-48, 2018.

\bibitem{guangwuxu}
G. Xu, ``On solving a generalized Chinese remainder theorem in the presence of remainder errors.'' In
{\em Conference on Geometry, Algebra, Number Theory, and their Information Technology Applications (GANITA)}, pp. 461-476, Springer, Cham, 2016.

\bibitem{yangyang}
B. Yang, W. J. Wang, X.-G. Xia, and Q. Yin, ``Phase detection based range estimation with a dual-band robust Chinese remainder theorem,''
{\em Sci. China Inf. Sci.}, vol. 57, no. 2, pp. 1-9, 2014.

\bibitem{yangyang2}
W. J. Wang, X. P. Li, W. Wang, and X.-G. Xia, ``Maximum likelihood estimation based robust Chinese remainder theorem for real numbers and its fast algorithm,''
{\em IEEE Trans. Signal Process.}, vol. 63, no. 13, pp. 3317-3331, 2015.

\bibitem{xiaoli1}
L. Xiao, X.-G. Xia, and W. J. Wang, ``Multi-stage robust Chinese remainder theorem,''
{\em IEEE Trans. Signal Process.}, vol. 62, no. 18, pp. 4772-4785, 2014.

\bibitem{xiaoli2}
L. Xiao, X.-G. Xia, and H. Y. Huo, ``Towards robustness in residue number systems,''
{\em IEEE Trans. Signal Process.}, vol. 65, no. 6, pp. 1497-1510, 2017.

\bibitem{optics1}
K. Falaggis, D. P. Towers, and C. E. Towers,
``Method of excess fractions with application to absolute distance metrology: Analytical solution,''
{\em Applied Optics}, vol. 52, no. 23, pp. 5758-5765, 2013.

\bibitem{optics2}
S. Tang, X. Zhang, and D. Tu,
``Micro-phase measuring profilometry: Its sensitivity analysis and phase unwrapping,''
{\em Opt. Lasers Eng.}, vol. 72, pp. 47-57, 2015.

\bibitem{optics3}
T. Petkovi\'{c}, T. Pribani\'{c}, and M. Donli\'{c},
``Temporal phase unwrapping using orthographic projection,''
{\em Opt. Lasers Eng.}, vol. 90, pp. 34-47, 2017.

\bibitem{wenchaoli}
W. C. Li, X. Z. Wang, and B. Moran,
``Wireless signal travel distance estimation using non-coprime wavelengths,''
{\em IEEE Signal Process. Lett.}, vol. 24, no. 1, pp. 27-31, 2017.

\bibitem{akhlaq}
A. Akhlaq, R. McKilliam, R. Subramanian, and A. Pollok,
``Selecting wavelengths for least squares range estimation,''
{\em IEEE Trans. Signal Process.}, vol. 64, no. 20, pp. 5205-5216, 2016.

\bibitem{ocean1}
C. Chi, H. Vishnu, K. T. Beng, and M. Chitre,
``Utilizing orthogonal coprime signals for improving broadband acoustic Doppler current profilers,''
{\em IEEE J. Oceanic Eng.}, DOI: 10.1109/JOE.2019.2925922.

\bibitem{ocean2}
C. Chi, H. Vishnu, K. T. Beng, and M. Chitre,
``Robust resolution of velocity ambiguity for multifrequency pulse-to-pulse coherent Doppler sonars,''
{\em IEEE J. Oceanic Eng.}, DOI: 10.1109/JOE.2019.2925919.

\bibitem{wsn1}
G. Campobello, A. Leonardi, and S. Palazzo,
``Improving energy saving and reliability in wireless sensor networks using
a simple CRT-based packet-forwarding solution,''
{\em IEEE/ACM Trans. Netw.}, vol. 20, no. 1, pp. 191-205, 2012.

\bibitem{wsn2}
S. Chessa and P. Maestrini,
``Robust distributed storage of residue encoded data,''
{\em IEEE Trans. Inf. Theory}, vol. 58, no. 12, pp. 7280-7294, 2012.

\bibitem{wsn3}
Y.-S. Su,
``Topology-transparent scheduling via the Chinese remainder theorem,''
{\em IEEE/ACM Trans. Netw.}, vol. 23, no. 5, pp. 1416-1429, 2015.

\bibitem{biology1}
I. Fiete, Y. Burak, and T. Brookings,
``What grid cells convey about rat location,''
{\em J. Neurosci.}, vol. 28, no. 27, pp. 6858-6871, 2008.

\bibitem{biology2}
M. Stemmler, A. Mathis, and A. V. M. Herz,
``Connecting multiple spatial scales to decode the population activity of grid cells,''
{\em Sci. Adv.}, vol. 1, no. 11, e1500816, 2015.

\bibitem{biology3}
Y. Yoo, O. O. Koyluoglu, S. Vishwanath, and I. Fiete,
``Multi-periodic neural coding for adaptive information transfer,''
{\em Theor. Comput. Sci.}, vol. 633, pp. 37-53, 2016.

\bibitem{lugan}
L. Gan and H. Liu,
``High dynamic range sensing using multi-channel modulo samplers.'' In
{\em IEEE 11th Sensor Array and Multichannel Signal Processing Workshop (SAM)}, 2020.

\bibitem{xialao1}
X.-G. Xia, ``On estimation of multiple frequencies in undersampled complex valued waveforms,''
{\em IEEE Trans. Signal Process.}, vol. 47, no. 12, pp. 3417-3419, 1999.

\bibitem{xialao2}
X.-G. Xia, ``An efficient frequency-determination algorithm from multiple undersampled waveforms,''
{\em IEEE Signal Process. Lett.}, vol. 7, no. 2, pp. 34-37, 2000.

\bibitem{xiaopingli1}
X. P. Li, X.-G. Xia, W. J. Wang, and W. Wang,
``A robust generalized Chinese remainder theorem for two integers,''
{\em IEEE Trans. Inf. Theory}, vol. 62, no. 12, pp. 7491-7504, 2016.

\bibitem{huiyongliao}
H. Liao and X.-G. Xia,
``A sharpened dynamic range of a generalized Chinese remainder theorem for multiple integers,''
{\em IEEE Trans. Inf. Theory}, vol. 53, no. 1, pp. 428-433, 2007.

\bibitem{xiaoli5}
L. Xiao, X.-G. Xia, and H. Y. Huo,
``New conditions on achieving the maximal possible dynamic range for a generalized Chinese remainder theorem of multiple integers,''
{\em IEEE Signal Process. Lett.}, vol. 22, no. 12, pp. 2199-2203, 2015.

\bibitem{hanshen1}
H. S. Xiao and G. Q. Xiao,
``On solving ambiguity resolution with robust Chinese remainder theorem for multiple numbers,''
{\em IEEE Trans. Veh. Technol.}, vol. 68, no. 5, pp. 5179-5184, 2019.

\bibitem{hanshen2}
H. S. Xiao, Y. F. Huang, Y. Ye, and G. Q. Xiao,
``Robustness in Chinese remainder theorem for multiple numbers and remainder coding,''
{\em IEEE Trans. Signal Process.}, vol. 66, no. 16, pp. 4347-4361, 2018.

\bibitem{xiaoli}
L. Xiao and X.-G. Xia,
``Frequency determination from truly sub-Nyquist samplers based on robust Chinese remainder theorem,''
{\em Signal Process.}, vol. 150, pp. 248-258, 2018.

\bibitem{jiawenxian}
A. Guessoum and R. Mersereau, ``Fast algorithms for the multidimensional discrete Fourier transform,''
{\em IEEE Trans. Acoust. Speech Signal Process.}, vol. 34, no. 4, pp. 937-943, 1986.

\bibitem{PPV1}
Y.-P. Lin, S.-M. Phoong, and P. P. Vaidyanathan, ``New results on multidimensional Chinese remainder theorem,''
{\em IEEE Signal Process. Lett.}, vol. 1, no. 11, pp. 176-178, 1994.

\bibitem{PPV3}
T. Chen and P. P. Vaidyanathan, ``Recent development in multidimensional multirate systems,''
{\em IEEE Trans. Circuits Syst. Video Technol.}, vol. 3, no. 2, pp. 116-137, 1993.

\bibitem{smith}
H. J. S. Smith, ``On systems of linear indeterminate equations and congruences,''
{\em Phil. Trans. Ray. Soc. London}, vol. 151, pp. 293-326, 1861.

\bibitem{PPV4}
P. P. Vaidyanathan and P. Pal, ``Theory of sparse coprime sensing in multiple dimensions,''
{\em IEEE Trans. Signal Process.}, vol. 59, no. 8, pp. 3592-3608, 2011.

\bibitem{lattice1}
C. C. MacDuffee,
{\em The Theory of Matrices},
New York: Chelsea, 1946.

\bibitem{smith3}
P. Pal and P. P. Vaidyanathan, ``Nested arrays in two dimensions, Part I: Geometrical considerations,''
{\em IEEE Trans. Signal Process.}, vol. 60, no. 9, pp. 4694-4705, 2012.

\bibitem{smith2}
T. Hori, ``Relationship between Smith normal form of periodicity matrices and sampling of two-dimensional discrete frequency distributions with tiling capability,''
{\em IEEE Trans. Circuits Syst. II, Exp. Briefs}, vol. 63, no. 2, pp. 191-195, 2016.

\bibitem{PPV5}
P. Pal and P. P. Vaidyanathan, ``Coprimality of certain families of integer matrices,''
{\em IEEE Trans. Signal Process.}, vol. 59, no. 4, pp. 1481-1490, 2011.

\bibitem{PPV6}
P. P. Vaidyanathan and P. Pal, ``A general approach to coprime pairs of matrices, based on minors,''
{\em IEEE Trans. Signal Process.}, vol. 59, no. 8, pp. 3536-3548, 2011.

\bibitem{latticproblem}
D. Micciancio and P. Voulgaris, ``A deterministic single exponential time algorithm for most lattice problems based on Voronoi cell computations,''
{\em SIAM J. Comput.}, vol. 42, no. 3, pp. 1364-1391, 2013.

\bibitem{latticproblem2}
G. Hanrot, X. Pujol, and D. Stehl\'{e}, ``Algorithms for the shortest and closest lattice vector problems.''
In {\em Coding and Cryptology}, Lecture Notes in Comput. Sci. 6639, pp. 159-190, Springer, Heidelberg, 2011.

\bibitem{enum}
C. P. Schnorr and M. Euchner, ``Lattice basis reduction: Improved practical algorithms and solving subset sum problems,''
{\em Math. Program.}, vol. 66, pp. 181-199, 1994.

\bibitem{boyd}
M. Grant and S. Boyd, ``CVX: Matlab software for disciplined convex programming, version 2.0 beta,''
{\em \url{http://cvxr.com/cvx}}, September 2013.

\bibitem{dft1}
D. Dudgeon and R. M. Mersereau,
{\em Multidimensional Digital Signal Processing},
Prentice Hal, 1984.

\bibitem{unitary}
P. Angeletti, ``Proof of unitarity of multidimensional discrete Fourier transform,''
{\em Electron. Lett.}, vol. 49, no. 7, pp. 501-503, 2013.

\bibitem{ppvbook}
P. P. Vaidyanathan, {\em Multirate Systems And Filter Banks},
Pearson Education India, 2006.


\end{thebibliography}
\end{document}